\documentclass[onecolumn, conference, final, letterpaper]{IEEEtran}

\usepackage[utf8]{inputenc}
\usepackage{graphicx}
\usepackage{amssymb}
\usepackage[cmex10]{amsmath}
\usepackage{amsthm} 
\usepackage{array}
\usepackage[usenames,dvipsnames]{xcolor}
\usepackage{hyperref}
\hypersetup{colorlinks=false, pdfborderstyle={/S/U/W 1},pdfborder=0 0 1, citebordercolor=Blue, filebordercolor=Red, linkbordercolor=Red, urlbordercolor=Blue}
\usepackage{xparse}
\usepackage{nicefrac}
\usepackage[english]{babel}
\usepackage[english = american]{csquotes}
\usepackage{tabularx}
\usepackage[ruled,vlined,titlenumbered]{algorithm2e}
\renewcommand{\algorithmcfname}{Algo}
\usepackage{enumitem}
\newtheorem{definition}{Definition}
\newtheorem{theorem}[definition]{Theorem}
\newtheorem{lemma}[definition]{Lemma}

\newtheorem{example}[definition]{Example}
\newtheorem{corollary}[definition]{Corollary}
\newtheorem{construction}[definition]{Construction}

\DeclareMathOperator{\defi}{def}
\newcommand{\defeq}{\overset{\defi}{=}}

\DeclareMathOperator{\rank}{rank}

\newcommand{\F}[1]{\mathbb F_{#1}}
\newcommand{\Fq}{\F{q}}

\newcommand{\N}{\ensuremath{\mathbb{N}}}
\newcommand{\Z}{\ensuremath{\mathbb{Z}}}
\renewcommand{\vec}[1]{\mathbf #1}
\newcommand{\M}[2][\empty]{
  \ifthenelse{\equal{#1}{\empty}}
    {\ensuremath{\mathbf{#2}}}
    {\ensuremath{{\mathbf{#2}}_{#1}}}
}

\newcommand{\CYC}{\ensuremath{\mathcal{C}}}

\newcommand{\LIN}[4]{\ensuremath{[#1,#2,#3]_{#4}}}
\newcommand{\inter}[1]{\ensuremath{[#1]}}
\newcommand{\defset}[2][\empty]{
  \ifthenelse{\equal{#1}{\empty}}
    {\ensuremath{D_{#2}}}
    {\ensuremath{D^{[#1]}_{#2}}}
}
\addto\captionsenglish{\renewcommand{\figurename}{Fig.}}
\newcommand{\loc}[1]{
  \ifthenelse{\equal{#1}{\empty}}
    {\ensuremath{r}}
    {\ensuremath{r_{#1}}}
}
\newcommand{\locset}[1]{\ensuremath{\mathcal{T}_{#1}}}
\newcommand{\minus}{\scalebox{0.5}[0.7]{\( - \)}}
\newcommand{\minusb}{\scalebox{0.6}[0.8]{\( - \)}}
\newcommand{\GCCOUT}[1]{\ifthenelse{\equal{#1}{\empty}}{\ensuremath{\mathcal{A}}}{\ensuremath{\mathcal{A}^{\langle #1 \rangle}}}}
\newcommand{\GCCOUTb}[1]{\ifthenelse{\equal{#1}{\empty}}{\ensuremath{\overline{\mathcal{A}}}}{\ensuremath{\overline{\mathcal{A}}^{\langle #1 \rangle}}}}
\newcommand{\GCCOUTn}[1]{\ifthenelse{\equal{#1}{\empty}}{\ensuremath{n_a}}{\ensuremath{n_{a,#1}}}}
\newcommand{\GCCOUTk}[1]{\ifthenelse{\equal{#1}{\empty}}{\ensuremath{k_a}}{\ensuremath{k_{a,#1}}}}
\newcommand{\GCCOUTd}[1]{\ifthenelse{\equal{#1}{\empty}}{\ensuremath{d_a}}{\ensuremath{d_{a,#1}}}}
\newcommand{\GCCOUTextq}[1]{\ifthenelse{\equal{#1}{\empty}}{\ensuremath{l}}{\ensuremath{l_{#1}}}}

\newcommand{\GCCIN}[1]{\ifthenelse{\equal{#1}{\empty}}{\ensuremath{\mathcal{B}}}{\ensuremath{\mathcal{B}^{\langle #1 \rangle}}}}
\newcommand{\GCCINn}[1]{\ifthenelse{\equal{#1}{\empty}}{\ensuremath{n_b}}{\ensuremath{n_{b,#1}}}}
\newcommand{\GCCINk}[1]{\ifthenelse{\equal{#1}{\empty}}{\ensuremath{k_b}}{\ensuremath{k_{b,#1}}}}
\newcommand{\GCCINd}[1]{\ifthenelse{\equal{#1}{\empty}}{\ensuremath{d_b}}{\ensuremath{d_{b,#1}}}}

\newcommand{\GCCn}{\ensuremath{n}}
\newcommand{\GCCk}{\ensuremath{k}}
\newcommand{\GCCd}{\ensuremath{d}}

\DeclareDocumentCommand \genmat { oo }
{
  \IfNoValueTF {#2}
  { % less then two values given
    \IfNoValueTF {#1}
    {   % NO parameter
      \ensuremath{\mathbf{G}}
    }
    { % one parameter given
      \ensuremath{\mathbf{G}_{#1}}    
    }
  }
  { % two values given
    \ensuremath{\overline{\mathbf{G}}_{#1}}           
  }
}
\newcommand{\CMoptimal}{alphabet-optimal}
\usepackage[backend=bibtex8, style=numeric, url=false, isbn=false, doi=true, maxnames=5, firstinits=true]{biblatex}
\usepackage[english]{babel}

\MakeOuterQuote{"}
\renewbibmacro{in:}{}
\DeclareListFormat{language}{}
\bibliography{MLLRCcomp}
\renewcommand{\bibfont}{\normalfont\footnotesize}
\begin{document}

\title{Bounds and Constructions of Codes with Multiple Localities}

\IEEEoverridecommandlockouts

\author{\IEEEauthorblockN{Alexander Zeh and Eitan Yaakobi}\thanks{This work has been supported by German Research Council (Deutsche Forschungsgemeinschaft, DFG) under grant ZE1016/1-1 and in part by the Israel Science Foundation (ISF) Grant No. 1624/14.}
\IEEEauthorblockA{Computer Science Department\\
Technion---Israel Institute of Technology\\
\texttt{alex@codingtheory.eu}, \texttt{yaakobi@cs.technion.ac.il}}
}

\maketitle

\begin{abstract}
This paper studies bounds and constructions of locally repairable codes (LRCs) with multiple localities so-called multiple-locality LRCs (ML-LRCs). In the simplest case of two localities some code symbols of an ML-LRC have a certain locality while the remaining code symbols have another one.

We extend two bounds, the Singleton and the alphabet-dependent upper bound on the dimension of Cadambe--Mazumdar for LRCs, to the case of ML-LRCs with more than two localities. Furthermore, we construct Singleton-optimal ML-LRCs as well as codes that achieve the extended alphabet-dependent bound. 
We give a family of binary ML-LRCs based on generalized code concatenation that is optimal with respect to the alphabet-dependent bound.
\end{abstract}
\begin{IEEEkeywords}
 Alphabet-dependent bound, distributed storage, generalized code concatenation, locally repairable codes, Pyramid code, Singleton bound, shortening
\end{IEEEkeywords}

\section{Introduction}
Locally repairable codes (LRCs) can recover data from erasure(s) by accessing a small number $\loc{}$ of erasure-free code symbols and therefore increase the efficiency of the repair-process in large-scale distributed storage systems.
For a classical LRC, all $n$ code symbols depend on at most $\loc{}$ other symbols. Basic properties and bounds of LRCs were identified by Gopalan~\textit{et al.}~\cite{gopalan_locality_2012}, Oggier and Datta~\cite{oggier_self-repairing_2011} and Papailiopoulos and Dimakis~\cite{papailiopoulos_locally_2014}. 
The majority of the constructions of LRCs requires a large field size (see, e.g.,~\cite{huang_pyramid_2007, tamo_optimal_2013, silberstein_optimal_2013}). Tamo and Barg gave a family of optimal LRCs for which the required field size is slightly larger than the code length in~\cite{tamo_family_2014}. Cadambe and Mazumdar~\cite{cadambe_upper_2013-1} gave an upper bound on the dimension of an $\loc{}$-local code which takes the field size into account. LRCs with small alphabet size were constructed in \cite{goparaju_binary_2014, huang_cyclic_2015,zeh_optimal_2015,silberstein_optimal_2015}.

We generalize the Singleton bound of Gopalan~\textit{et al.}~\cite{gopalan_locality_2012} as well as the alphabet-dependent Cadambe--Mazumdar bound~\cite{cadambe_upper_2013-1} to so-called multiple-locality LRCs (ML-LRCs), i.e., LRCs where the locality among the code symbols can be different. There exist several practical scenarios, where different localities are relevant, e.g., a distributed storage system, where a part of the stored data is accessed more frequently and therefore requires a smaller locality.

We show that shortening a given Singleton- respectively \CMoptimal{} LRC (i.e., a LRC that achieves the corresponding bound with equality and where all $n$ symbols have the same locality), a Singleton- respectively \CMoptimal{} ML-LRC is obtained. For the case of two localities, we give explicit algorithms that return a parity-check matrix of an optimal linear code for both bounds. Both algorithms can be easily extended to the general case of more than two localities. The adaption of the Pyramid code construction~\cite{huang_pyramid_2007} to the case of multiple localities is outlined, because they require a (slightly) smaller field-size as the construction of~\cite{tamo_family_2014} even though the locality-property is only guaranteed for the information symbols). Furthermore, we identify a family of binary \CMoptimal{} ML-LRCs using generalized code concatenation (see e.g., \cite{blokh_coding_1974, zinoviev_generalized_1976}).

This paper is structured as follows. In Section~\ref{sec_Preliminaries}, we introduce notation and recall the existing Singleton and alphabet-dependent bound of Cadambe--Mazumdar for LRCs (Thm.~\ref{theo_GenSingleton} and Thm.~\ref{theo_CMBound}).
Section~\ref{sec_LocalityCode} provides the definition of ML-LRCs and the Singleton bound for the case of $s=2$ different localities (Thm.~\ref{theo_SingletonDifferentLocalities}). Thm.~\ref{theo_SingletonDifferentLocalitiesSeveral} is the Singleton bound for the general case of linear ML-LRCs with $s \geq 2$ different localities. The proof is given in the appendix. An explicit construction (Algorithm~\ref{algo_ConstructTwoDiffLocalitiesLRC}) of a Singleton-optimal ML-LRC with two different localities is proven in Thm.~\ref{theo_SingletonOptimalMLLRC}. Basically, Algorithm~\ref{algo_ConstructTwoDiffLocalitiesLRC} shortens a given Singleton-optimal LRC code. The adaption of Pyramid codes~\cite{huang_pyramid_2007} for different localities among the information symbols is outlined at the end of Section~\ref{sec_LocalityCode}. Similar to Section~\ref{sec_LocalityCode}, we start the description of the extended alphabet-dependent bound (for non-linear) LRCs in Thm.~\ref{thm_AlphabetDiffLocalities} in Section~\ref{sec_Alphabet-DependentBound} for two different localities. The general case ($s \geq 2$) is contained in Thm.~\ref{thm_AlphabetDiffLocalitiesSeveral} and the proof can be found in the appendix. Our proof requires linearity, but it is straightforward to extend the proof-idea to the non-linear case similar to the graph approach as in Tamo--Barg~\cite[Proof of Thm. A.1]{tamo_family_2014}. Thm.~\ref{thm_ConstructionOptimalMLLRC} proves an explicit construction (Algorithm~\ref{algo_ConstructTwoDiffLocalitiesLRCAlphabetDependent}) of an ML-LRC that is optimal with respect to new alphabet-dependent bound (for two different localities). Similar to Algorithm~\ref{algo_ConstructTwoDiffLocalitiesLRC}, Algorithm~\ref{algo_ConstructTwoDiffLocalitiesLRCAlphabetDependent} shortens a given \CMoptimal{} LRC. Both algorithms (Algorithm~\ref{algo_ConstructTwoDiffLocalitiesLRC} and \ref{algo_ConstructTwoDiffLocalitiesLRCAlphabetDependent}) can be adapted to the case of an ML-LRC with more than two different localities.

Section~\ref{sec_OptimalBinaryGCC} contains a family of $\loc{}$-local binary LRCs based on generalized code concatenation. It is shown that this family is optimal with respect to the Cadambe--Mazumdar bound for $\loc{}=2,3$ (see Construction~\ref{constr_GCC1} and Thm.~\ref{theo_OptimalityConstrGCC2}). We conclude the paper in Section~\ref{sec_Concl}.

\section{Preliminaries} \label{sec_Preliminaries}
For two integers $a,b$ with $a < b$, let $\inter{a,b}$ denote the set of integers $\{a,a+1,\dots,b\}$ and let $\inter{b}$ be the shorthand notation for $\inter{1,b}$.
Let $\Fq$ denote the finite field of order $q$, where $q$ is a prime power. A linear $\LIN{n}{k}{d}{q}$ code of length $n$, dimension $k$ and minimum Hamming distance $d$ over $\Fq$ is denoted by a calligraphic letter like $\CYC$. The parameters of a non-linear code $\CYC$ of length $n$, over an alphabet with size $q$ with dimension $k = \log |\CYC| / \log q$ are identified by $(n,k,d)_q$.
For a set $\mathcal{I} \subseteq \inter{n}$ and a matrix $\mathbf{G} \in \Fq^{k \times n}$, let $\mathbf{G}_{\mathcal{I}}$ denote the submatrix of $\mathbf{G}$ that consists of the columns indexed by $\mathcal{I}$. Denote by $\mathbf{I}_{n}$ the $n \times n$ identity matrix.
\begin{definition}[Locally Repairable Code (LRC)] \label{def_Locality}
A linear \LIN{n}{k}{d}{q} code $\CYC$ is $\loc{}$-local if all $n$ code symbols are a linear combination of at most $\loc{}$ other code symbols.
\end{definition}
The following generalization of the Singleton bound for LRCs was among others proven in~\cite[Thm. 3.1]{kamath_codes_2014}, \cite[Construction 8 and Thm. 5.4]{tamo_family_2014} and \cite[Thm. 2]{prakash_optimal_2012}.
\begin{theorem}[Generalized Singleton Bound] \label{theo_GenSingleton}
The minimum Hamming distance $d$ of an \LIN{n}{k}{d}{q} $r$-local code $\CYC$ is bounded from above by
\begin{equation*} 
d \leq n-k+2-\left \lceil \frac{k}{r} \right \rceil.
\end{equation*}
\end{theorem}
Throughout this paper we call an $\loc{}$-local code \textit{Singleton-optimal} if its minimum Hamming distance achieves the bound in Thm.~\ref{theo_GenSingleton} with equality.
For $r \geq k $ Thm.~\ref{theo_GenSingleton} coincides with the classical (locality-unaware) Singleton bound and then a Singleton-optimal code is called maximum distance separable (MDS).

In addition to the generalization of the bound in Thm.~\ref{theo_GenSingleton}, we extend the bound on the dimension of an LRC given by Cadambe and Mazumdar~\cite[Thm. 1]{cadambe_upper_2013-1}, which takes the alphabet into account.
\begin{theorem}[Cadambe--Mazumdar Bound] \label{theo_CMBound}
The dimension $k$ of an $r$-local code $\CYC$ of length $n$ and minimum Hamming distance $d$ as in Definition~\ref{def_Locality} is upper bounded by 
\begin{equation} \label{eq_CMBounda}
k \leq \min_{t \in \inter{t^{\star}}} k_t ,
\end{equation}
where 
\begin{equation} \label{eq_CMBoundb}
k_t \defeq  tr + k_{\text{opt}}^{(q)}(n-t(\loc{}+1), d),
\end{equation}
where $k_{opt}^{(q)}(n, d)$ is the largest possible dimension of a code of length $n$, for a given alphabet size $q$ and a given minimum Hamming distance $d$ and $t^{\star} = \min \left( \left \lceil n/(\loc{}+1) \right \rceil, \left \lceil k/\loc{} \right \rceil \right)$.
\end{theorem}
An $\loc{}$-local code is called \textit{\CMoptimal{}} if its dimension meets the bound in Thm.~\ref{theo_CMBound} with equality.

\section{Multiple-Locality LRCs} \label{sec_LocalityCode}
In this section we first define ML-LRCs with two localities and provide a Singleton-like bound. We generalize the definition to ML-LRCs with more than two localities and extend the bound. Afterwards, Singleton-optimal all-symbol ML-LRCs are constructed and the construction of Pyramid codes~\cite{huang_pyramid_2007} with information-symbol locality is extended.

For a linear code of length $n$, we consider in a first step the case where all code symbols in $\locset{1} \subset \inter{n}$ have locality $\loc{1}$ and the remaining code symbols in $\locset{2} = \inter{n} \setminus \locset{1}$ have locality $\loc{2}$.
\begin{definition}[ML-LRC with Two Localities] \label{def_CodeDiffLocalities}
Let $\locset{1} \subset \inter{n}$ and $\locset{2} = \inter{n} \setminus \locset{1}$ be two distinct sets with $ |\locset{i}| = n_i,$ for $i=1,2$. Let two integers $\loc{1},\loc{2} $ with $\loc{1} < \loc{2}$ be given.
A linear $\LIN{n}{k}{d}{q}$ code is $((n_1, \loc{1}), (n_2, \loc{2}))$-local if all $n_i$ code symbols within $\locset{i}$ are a linear combination of at most $\loc{i}$ other code symbols for $i=1,2$.
\end{definition}
Note, that a code symbol belongs to $\locset{1}$ if it is a linear combination of at most $\loc{1}$ other symbols. Clearly, an $((n_1, \loc{1}), (n_2, \loc{2}))$-local code is also an $\loc{2}$-local code and an $\loc{1}$-local code is an $((n_1, \loc{1}), (n_2, \loc{2}))$-local code.
\subsection{Singleton Bound}
The following lemma is needed to prove the Singleton-like bound on the minimum Hamming distance of an $((n_1, \loc{1}), (n_2, \loc{2}))$-local code.
\begin{lemma}[Rank of Generator Matrix] \label{lem_RankGeneratorMatrix}
Let $\mathcal{I} \subseteq \inter{n}$. All $k \times |\mathcal{I}|$ submatrices $\mathbf{G}_{\mathcal{I}}$ of a generator matrix $\mathbf{G}$ of an $\LIN{n}{k}{d}{q}$ code $\CYC$ of rank smaller than $k$ must have at most $n-d$ columns. 
\end{lemma}
\begin{proof}
W.l.o.g. we assume $\mathcal{I} = \inter{|\mathcal{I}|}$ and let $t = \rank \mathbf{G}_{\mathcal{I}}$. Clearly, $|\mathcal{I}| \geq t$. The generator matrix can be transformed to the following form:
\begin{equation*}
\mathbf{G}= 
\begin{pmatrix}  
\mathbf{I}_{t} & \mathbf{A}_1  & \mathbf{A}_2 \\ 
\multicolumn{2}{c}{\mathbf{0}} & \mathbf{A}_3 
\end{pmatrix},
\end{equation*}
where $\mathbf{A}_1 \in \Fq^{t \times (|\mathcal{I}|-t)}$, $\mathbf{A}_2 \in \Fq^{t \times (n-|\mathcal{I}|)}$, and $\mathbf{A}_3 \in \Fq^{(k-t) \times (n-|\mathcal{I}|)}$. The $(t+1)$th row of $\mathbf{G}$ (first $|\mathcal{I}|$ positions are zero) has Hamming weight at least $d$ (because it is a codeword of $\CYC$) and therefore $|\mathcal{I}| \leq n-d$.
\end{proof}

\begin{theorem}[Singleton Bound for ML-LRCs with Two Localities] \label{theo_SingletonDifferentLocalities}
Let the parameters as in Definition~\ref{def_CodeDiffLocalities} be given.

Then, the minimum Hamming distance of an $\LIN{n}{k}{d}{q}$ $((n_1, \loc{1}), (n_2, \loc{2}))$-local code $\CYC$ is upper bounded by
\begin{align} \label{eq_SingletonExp1}
d & \leq  n - k + 2 - \left \lceil \frac{n_1}{\loc{1}+1} \right \rceil - \left \lceil \frac{k- \loc{1}\left \lceil \frac{n_1}{\loc{1}+1} \right \rceil }{\loc{2}} \right \rceil.
\end{align}
\end{theorem}
\begin{proof}
The proof follows the idea of the proof of the Singleton bound of an $\loc{}$-local code as in~\cite{kuijper_erasure_2014}. Let $\mathbf{G}$ be a $k \times n$ generator matrix of $\CYC$. 
We assume that $\loc{1} \lceil n_1/(\loc{1}+1) \rceil < k-1$. 
\begin{enumerate}[leftmargin=.5cm]
\item Choose $\kappa_1 \defeq \lceil n_1/(\loc{1}+1) \rceil$ columns of $\mathbf{G}$ indexed by $j_1^{\langle 1\rangle},j_2^{\langle 1\rangle},\dots,j_{\kappa_1}^{\langle 1\rangle}$, where $j_i^{\langle 1\rangle} \in \locset{1}$. Each column is a linear combination of at most $\loc{1}$ other columns. 
\item Let $\mathcal{I}^{\langle 1 \rangle}$ be the set of indexes of all repair columns of the columns indexed by  $j_1^{\langle 1\rangle},\dots,j_{\kappa_1}^{\langle 1\rangle}$, but without the indexes themselves. We have $| \mathcal{I}^{\langle 1 \rangle}| \leq \loc{1} \cdot \kappa_1 < k-1$.
\item Choose 
\begin{equation*}
\kappa_2 \defeq \left \lfloor \frac{k-1- \loc{1}\left \lceil \frac{n_1}{\loc{1}+1} \right \rceil}{\loc{2}} \right \rfloor 
\end{equation*}
columns of $\mathbf{G}$ indexed by $j_1^{\langle 2\rangle},j_2^{\langle 2\rangle},\dots,j_{\kappa_2}^{\langle 2\rangle}$, where $j_i^{\langle 2\rangle} \in \locset{2}$. Now, each of these columns is a linear combination of at most $\loc{2}$ other columns. Let $\mathcal{I}^{\langle 2 \rangle}$ %\subset \locset{2}$
be the set of indexes of all repair columns of the columns indexed by $j_1^{\langle 2\rangle},j_2^{\langle 2\rangle},\dots,j_{\kappa_2}^{\langle 2\rangle}$, but without the indexes themselves.
\item Let $ \mathcal{I} \defeq \mathcal{I}^{\langle 1 \rangle} \cup \mathcal{I}^{\langle 2 \rangle}$. Then $| \mathcal{I}| < k$ and we have $\rank(\mathbf{G}_{\mathcal{I}}) < k$. 
\item Enlarge $\mathcal{I}$ to a set $\mathcal{I}'$, such that $\rank(\mathbf{G}_{\mathcal{I}'}) = k-1$, but without using $\left\{j_1^{\langle i \rangle},j_2^{\langle i \rangle},\dots,j_{\kappa_i}^{\langle i \rangle}\right\}_{i=1,2}$. %\textcolor{red}{why possible?}
\item Then, define 
\begin{equation*}
\begin{split}
\mathcal{U} \defeq \left\{\mathcal{I}' \cup \left\{ j_1^{\langle 1 \rangle},j_2^{\langle 1 \rangle},\dots,j_{\kappa_1}^{\langle 1 \rangle}, j_1^{\langle 2 \rangle},j_2^{\langle 2 \rangle},\dots,j_{\kappa_2}^{\langle 2 \rangle}\right\} \right\}, 
\end{split}
\end{equation*}
where $| \mathcal{U}| \geq k-1 + \kappa_1 + \kappa_2$. 
We have $\rank(\mathbf{G}_{\mathcal{U}}) \leq k-1$. 
\end{enumerate}
Using Lemma~\ref{lem_RankGeneratorMatrix}, we know that $|\mathcal{U}|$ can be upper bounded by $k-1 + \kappa_1 + \kappa_2 \leq n-d$ and therefore we obtain the following bound on the minimum Hamming distance:
\begin{align*}
d & \leq n - k + 1 - \left \lceil \frac{n_1}{\loc{1}+1} \right \rceil - \left \lfloor \frac{k-1- \loc{1}\left \lceil \frac{n_1}{\loc{1}+1} \right \rceil }{\loc{2}} \right \rfloor \\
& = n - k + 2 - \left \lceil \frac{n_1}{\loc{1}+1} \right \rceil - \left \lceil \frac{k- \loc{1}\left \lceil \frac{n_1}{\loc{1}+1} \right \rceil }{\loc{2}} \right \rceil.
\end{align*}
In case $\loc{1} \lceil n_1/(\loc{1}+1) \rceil \geq k-1$, set $\kappa_1 = \lfloor (k-1)/\loc{1} \rfloor $ in Step 1) and $\kappa_2 = 0$ in Step 2). Then, we will obtain the Singleton bound of an $\loc{1}$-local code.
\end{proof}
For $\loc{1}=\loc{2}=\loc{}$, we obtain from~\eqref{eq_SingletonExp1} the Singleton bound of an $\loc{}$-local LRC as in Thm.~\ref{theo_GenSingleton}.

We extend Definition~\ref{def_CodeDiffLocalities} to ML-LRCs with $s \geq 2$ localities and generalize the Singleton-like bound on the minimum Hamming distance from Thm.~\ref{theo_SingletonDifferentLocalities}.
\begin{definition}[ML-LRC] \label{def_MultipleLRC}
Let $s$ integers $\loc{1}, \loc{2}, \dots, \loc{s}$ with $\loc{1} < \loc{2} < \dots < \loc{s}$ be given.
Denote by $\locset{i} \subset \inter{n}$ for all $i \in \inter{s}$ pairwise disjoint sets. Let $\cup_{i \in \inter{s}} \locset{i} = \inter{n}$.  
A linear $\LIN{n}{k}{d}{q}$ code is $((n_1, \loc{1}), (n_2, \loc{2}), \dots , (n_{s}, \loc{s}))$-local if all $n_i$ code symbols within a set $\locset{i}$ with $|\locset{i}| = n_i$ are a linear combination of at most $\loc{i}$ other code symbols within $\locset{i}$ for all $i \in \inter{s}$. 
\end{definition}
Now, we give a Singleton-like bound for an $((n_1, \loc{1}), (n_2, \loc{2}),\dots , (n_{s}, \loc{s}))$-local code as in Definition~\ref{def_MultipleLRC}.
\begin{theorem}[Singleton Bound for ML-LRCs] \label{theo_SingletonDifferentLocalitiesSeveral}
Let $\CYC$ be an $\LIN{n}{k}{d}{q}$ $((n_1, \loc{1}), (n_2, \loc{2}), \dots , (n_{s}, \loc{s}))$-local code as in Definition~\ref{def_MultipleLRC}. Then, the minimum Hamming distance of $\CYC$ is upper bounded by
\begin{align} \label{eq_SingletonBoundMultipleLocalities}
d & \leq n \minusb k + 2 \minusb \sum_{i \in \inter{s \minus 1}} \left \lceil \frac{n_i}{\loc{i}+1} \right \rceil \minusb \left \lceil \frac{k \minusb \sum \limits_{i \in \inter{s \minus 1}} \loc{i}\left \lceil \frac{n_i}{\loc{i}+1} \right \rceil }{\loc{s}} \right \rceil.
\end{align}
\end{theorem}
\begin{proof}
See appendix.
\end{proof}
We call an $((n_1, \loc{1}), (n_2, \loc{2}),\dots , (n_{s}, \loc{s}))$-local ML-LRC with minimum Hamming distance that fulfills~\eqref{eq_SingletonBoundMultipleLocalities} with equality \textit{Singleton-optimal}.

\subsection{Shortening and Constructions}
In this subsection, we first show that an $((n_1, \loc{1}), (n_2, \loc{2}))$-local Singleton-optimal ML-LRC can be obtained through shortening an $r$-local Singleton-optimal LRC. Then, we analyze the shortening of an $((n_1, \loc{1}), (n_2, \loc{2}), \dots , (n_{s}, \loc{s}))$-local Singleton-optimal ML-LRC. We give an explicit construction of an $((n_1, \loc{1}), (n_2, \loc{2}))$-local Singleton-optimal code for the ease of notation. The construction can be easily extended to $s \geq 2$ localities. Furthermore, we adapt the construction of information-symbol-local Pyramid codes to the case of multiple localities.

We consider the case of shortening the $i$th information symbol (see, e.g., \cite[Ch. 18 \S 9]{macwilliams_theory_1988}).
\begin{definition}[Shortening] \label{def_Shortening}
Let $\CYC$ be an $\LIN{n}{k}{d}{q}$ code with generator matrix in systematic form, i.e., $\mathbf{G} = (\mathbf{I}_k \ \mathbf{A})$. A parity-check matrix is then $\mathbf{H} = (-\mathbf{A}^T \ \mathbf{I}_{n-k} )$. The shortened code is defined as 
\begin{equation*}
\begin{split}
\CYC^{\langle i \rangle} \defeq & \lbrace (c_1 \ c_2 \ \dots \ c_{i-1} \ c_{i+1} \ \dots \ \ c_{n}) : \\
 & \qquad (c_1 \ c_2 \ \dots \ c_{i-1} \ 0 \ c_{i+1} \ \dots \ \ c_{n}) \in \CYC \rbrace.
\end{split}
\end{equation*}
For $i \in \inter{k}$, the generator matrix of $\CYC^{\langle i \rangle}$ is obtained through deleting the $i$th column of $\mathbf{I}_k$ (and the corresponding $i$th row) of $\mathbf{G}$. A parity-check of $\CYC^{\langle i \rangle}$ is obtained by deleting the $i$th column in $\mathbf{A}^T$ of $\mathbf{H}$. The shortened code is an $\LIN{n-1}{k-1}{\geq d}{q}$ code. 
\end{definition}
Throughout this paper we refer shortening to the case where $i \in \inter{k}$.

Clearly, if $\CYC$ is an MDS code, then the shortened code $\CYC^{\langle i \rangle}$ is also an MDS code. The following lemma shows that this holds similarly for an $\loc{}$-local Singleton-optimal LRC.
\begin{lemma}[Shortening an $\loc{}$-local LRC] \label{lem_shorteningLRC}
Let an $\LIN{n}{k}{d}{q}$ $\loc{}$-local Singleton-optimal LRC be given. Then, the shortened $\LIN{n-1}{k-1}{d}{q}$ code is an $((\loc{}, \loc{}-1), (n-1-\loc, \loc{}))$-local code and is Singleton-optimal.
\end{lemma}
\begin{proof}
Clearly, the shortening by one symbol affects at least one repair-set and its cardinality is decreased to $\loc{}$. The locality of the contained code symbols is then $\loc{}-1$. We obtain from~\eqref{eq_SingletonExp1} with $n_1=\loc{}$, $\loc{1} = \loc{}-1$ and $\loc{2}=\loc{}$ for the shortened code of length $n-1$ and dimension $k-1$ the following bound on the minimum Hamming distance:
\begin{align*}
d & \leq n \minusb 1 \minusb (k \minusb 1) + 2 \minusb \left \lceil \frac{n_1}{\loc{1}+1} \right \rceil \minusb \left \lceil \frac{ (k \minusb 1) \minusb \loc{1}\left \lceil \frac{n_1}{\loc{1}+1} \right \rceil }{\loc{2}} \right \rceil \\
& = n - k + 2 - \left \lceil \frac{\loc{}}{\loc{}} \right \rceil - \left \lceil \frac{k-1- (\loc{}-1)\left \lceil \frac{\loc{}}{\loc{}} \right \rceil }{\loc{}} \right \rceil \\
& = n - k + 2 - \left \lceil \frac{k}{\loc{}} \right \rceil, 
\end{align*}
which coincides with the minimum Hamming distance of the given $\LIN{n}{k}{d}{q}$ $\loc{}$-local Singleton-optimal code.
\end{proof}
\begin{example}[Shortened Singleton-optimal LRC]
We shorten the $\LIN{12}{6}{6}{13}$ $3$-local Singleton-optimal code of Example 2 in~\cite{tamo_family_2014} by one position. We obtain an $\LIN{11}{5}{6}{13}$ $((3, 2), (8, 3))$-local code and Thm.~\ref{theo_SingletonDifferentLocalities} gives $d \leq 11 - 5 +2 - \lceil 3 /(2+1) \rceil - \lceil (5-2) /3 \rceil = 6$.  
\end{example}
\begin{lemma}[Shortening an ML-LRC] \label{lem_ShorteningSeveral}
Let $\CYC$ be an $\LIN{n}{k}{d}{q}$ $((n_1, \loc{1}), (n_2, \loc{2}), \dots , (n_{s}, \loc{s}))$-local Singleton-optimal code. Let $\locset{1}, \locset{2}, \dots, \locset{s}$ denote the locality sets and let $\alpha \in \inter{s}$. Shortening $\CYC$ by a coordinate that is contained in $\locset{\alpha}$ gives if $\loc{\alpha-1} =\loc{\alpha}-1$ an $\LIN{n'=n-1}{k'=k-1}{d}{q}$ $((n'_1, \loc{1}'), (n'_2, \loc{2}'),\dots , (n_{s}', \loc{s}'))$-local code $\CYC'$ with
\begin{equation} \label{eq_DiffLocShortCaseOne}
\begin{split}
(n'_i, \loc{i}') & = (n_i, \loc{i}), \qquad \forall i \in \inter{s} \setminus \{\alpha-1, \alpha \}, \\
(n'_{\alpha-1}, \loc{\alpha-1}') & = (n_{\alpha-1}+\loc{\alpha}, \loc{\alpha-1}), \\
(n'_{\alpha}, \loc{\alpha}') & = (n_{\alpha} - \loc{\alpha}-1, \loc{\alpha}), 
\end{split}
\end{equation}
else (i.e., $\loc{\alpha-1} \neq \loc{\alpha}-1$) shortening gives an $\LIN{n'=n-1}{k'=k-1}{d}{q}$ $((n'_1, \loc{1}'), (n'_2, \loc{2}'),\dots, (n'_{\iota}, \loc{\iota}'), \dots , (n_{s}', \loc{s}'))$-local code $\CYC'$ with
\begin{equation} \label{eq_DiffLocShortCaseTwo}
\begin{split}
(n'_i, \loc{i}') & = (n_i, \loc{i}), \qquad \forall i \in \inter{s} \setminus \{\alpha\}, \\
(n'_{\iota}, \loc{\iota}') & = (\loc{\alpha}, \loc{\alpha}-1),\\
(n'_{\alpha}, \loc{\alpha}') & = (n_{\alpha} - \loc{\alpha}-1, \loc{\alpha}).
\end{split}
\end{equation}
Then, the shortened $\LIN{n'=n-1}{k'=k-1}{d}{q}$ $((n'_1, \loc{1}'),(n'_2, \loc{2}'), \dots , (n_{s}', \loc{s}'))$-local respectively the $((n'_1, \loc{1}'),(n'_2, \loc{2}'), $ ... $, (n'_{\iota}, \loc{\iota}'), $ ...  $ , (n_{s}', \loc{s}'))$-local ML-LRC is Singleton-optimal.
\end{lemma}
\begin{proof}
See appendix.
\end{proof}
Algorithm~\ref{algo_ConstructTwoDiffLocalitiesLRC} provides a parity-check matrix of a Singleton-optimal $\LIN{n_1 + n_2}{k}{d}{q}$ $ ((n_1, \loc{1}), (n_2, \loc{2}))$-local code. 
\begin{center}
%\vspace{-.3cm}
\begin{algorithm}[h]
\label{algo_ConstructTwoDiffLocalitiesLRC}
%\IncMargin{1em}
\SetAlgoVlined
\DontPrintSemicolon
\LinesNumbered
%\Indm  
\SetKwInput{KwSub}{Subroutines}
\SetKwInput{KwIn}{Input}
\SetKwInput{KwOut}{Output}
\newcommand\mycommfont[1]{\footnotesize\ttfamily{#1}}
\SetCommentSty{mycommfont}
\BlankLine
\KwIn{\\Parity-check matrix $\mathbf{H}$ of an $\LIN{n_2 + (\loc{2}+1)\frac{n_1}{\loc{1}+1}}{k + (\loc{2}-\loc{1})\frac{n_1}{\loc{1}+1} }{d}{q}$ $\loc{2}$-local Singleton-optimal code with $\rho = \frac{n_2}{\loc{2}+1} + \frac{n_1}{\loc{1}+1}$ repair sets $\mathcal{R}_1, \mathcal{R}_2, \dots, \mathcal{R}_{\rho}$.
}
\BlankLine 
%\Indp
\BlankLine 
\For{$i \in \inter{\frac{n_1}{\loc{1}+1}}$}
{
Delete $\loc{2}-\loc{1}$ columns of $\mathbf{H}$ that are contained in $\mathcal{R}_i$.
}
\BlankLine
\BlankLine
\KwOut{Parity-check matrix $\mathbf{H}$ of an $\LIN{n_1 + n_2}{k}{d}{q}$ $ ((n_1, \loc{1}), (n_2, \loc{2}))$-local code.}
\caption{Singleton-optimal ML-LRC with $r_2 > r_1$.}
\end{algorithm}
%\vspace{-.7cm}
\end{center}
We assume that $(\loc{1}+1) \mid n_1$, $(\loc{2}+1) \mid n_2$, and that an $\loc{2}$-local Singleton-optimal code is given.
\begin{theorem}[Singleton-optimal ML-LRC] \label{theo_SingletonOptimalMLLRC}
Algorithm~\ref{algo_ConstructTwoDiffLocalitiesLRC} returns a parity-check matrix of a Singleton-optimal $\LIN{n_1 + n_2}{k}{d}{q}$ $ ((n_1, \loc{1}), (n_2, \loc{2}))$-local code that covers the possible rate-regime.
\end{theorem}
\begin{proof}
The minimum Hamming distance of the given $\LIN{n'}{k'}{d'}{q}$ $\loc{2}$-local Singleton-optimal LRC, where
\begin{align}
n' & = n_2 + (\loc{2}+1)\frac{n_1}{\loc{1}+1},  \label{eq_SingletonMLLRClength}\\
k' & = k + (\loc{2}-\loc{1})\frac{n_1}{\loc{1}+1}, \label{eq_SingletonMLLRCdimension}
\end{align}
equals
\begin{align} \label{eq_DistanceMLLRCConstruct}
d' & = n' - k' + 2 - \left \lceil \frac{k'}{\loc{2}} \right \rceil.
\end{align}
Inserting the expression of the length and the dimension as in~\eqref{eq_SingletonMLLRClength} respectively~\eqref{eq_SingletonMLLRCdimension} into~\eqref{eq_DistanceMLLRCConstruct} leads to:
\begin{align}
d' & \leq n_2 + (\loc{2}+1)\frac{n_1}{\loc{1}+1} - k - (\loc{2}-\loc{1})\frac{n_1}{\loc{1}+1} + 2 \nonumber \\
& \qquad - \left \lceil \frac{k + (\loc{2}-\loc{1})\frac{n_1}{\loc{1}+1}}{\loc{2}} \right \rceil \nonumber \\
& = n_2 + \frac{n_1}{\loc{1}+1} - k + \loc{1}\frac{n_1}{\loc{1}+1} + 2 - \left \lceil \frac{k -\loc{1}\frac{n_1}{\loc{1}+1}}{\loc{2}} \right \rceil \nonumber \\
& \qquad - \frac{n_1}{\loc{1}+1} \nonumber \\
& = n_2 - k + \loc{1}\frac{n_1}{\loc{1}+1} + 2 - \left \lceil \frac{k -\loc{1} \frac{n_1}{\loc{1}+1}}{\loc{2}} \right \rceil \nonumber \\
& = n_2 - k + (\loc{1}+1)\frac{n_1}{\loc{1}+1} + 2 - \frac{n_1}{\loc{1}+1} - \left \lceil \frac{k -\loc{1}\frac{n_1}{\loc{1}+1}}{\loc{2}} \right \rceil \nonumber \\
& =  n_1 + n_2 - k + 2 - \frac{n_1}{\loc{1}+1} - \left \lceil \frac{k -\loc{1}\frac{n_1}{\loc{1}+1}}{\loc{2}} \right \rceil, 
\end{align}
which is the minimum Hamming distance of an $\LIN{n_1 + n_2}{k}{d}{q}$ $ ((n_1, \loc{1}), (n_2, \loc{2}))$-local code that is optimal with respect to the bound given in Thm.~\ref{theo_SingletonDifferentLocalities}.

To prove the achievable rate-regime, we recall that the original $\loc{2}$-local LRC has to satisfy:
\begin{equation} \label{eq_RateRestriction}
k' \leq \frac{\loc{2}}{\loc{2}+1} n',
\end{equation}
and inserting the expressions of~\eqref{eq_SingletonMLLRClength} and~\eqref{eq_SingletonMLLRCdimension} in~\eqref{eq_RateRestriction} gives 
\begin{align*}
k + (\loc{2}-\loc{1})\frac{n_1}{\loc{1}+1} & \leq  \frac{\loc{2}}{\loc{2}+1} n_2 + (\loc{2}+1)\frac{n_1}{\loc{1}+1} \\
\Leftrightarrow k & \leq \frac{\loc{1}}{\loc{1}+1}n_1 + \frac{\loc{2}}{\loc{2}+1}n_2,
 \end{align*}
which is clearly the rate-restriction for an $ ((n_1, \loc{1}), (n_2, \loc{2}))$-local code.
\end{proof}
Pyramid codes~\cite{huang_pyramid_2007} are a class of Singleton-optimal \textit{information-symbol}-local codes. The construction is based on an MDS code and the maximal code length of a Pyramid code scales linearly with the field size.
Let $\mathbf{a} \cdot \mathbf{b} =  \sum_{i \in \inter{n}} a_i b_i$ denote the dot-product for two given vectors $\mathbf{a}, \mathbf{b} \in \Fq^{n}$. 
\begin{construction}[ML-Pyramid Code] \label{constr_PyramidCode}
Let $s$ integers $\loc{1}, \loc{2}, \dots, \loc{s}$ with $\loc{1} < \loc{2} < \dots < \loc{s}$ be given.
Denote by $\locset{i} \subseteq \inter{k}$ with $k_i \defeq |\locset{i}| $ for all $i \in \inter{s}$ pairwise disjoint sets and let $\cup_{i \in \inter{s}} \locset{i} = \inter{k}$.
We partition $\locset{i}$ into $\kappa_i = \lceil k_i/ \loc{i} \rceil$ disjoint subsets $\mathcal{I}_{i,1},\mathcal{I}_{i,2},\dots, \mathcal{I}_{i,\kappa_i}$, such that $\locset{i} = \cup_{j \in \inter{\kappa_i}} \mathcal{I}_{i,j}$ and $|\mathcal{I}_{i,j}| \leq \loc{i}$ for all $i \in \inter{s}$. 
Let a $\LIN{k+d-1}{k}{d}{q}$ MDS code $\CYC$ and an information vector $\mathbf{x} \in \F{q}^{k}$ be given.
Any codeword $\mathbf{c} \in \CYC$ is of the following form:
\begin{equation*} 
\mathbf{c} = \left(\mathbf{x}, \mathbf{p}^{\langle 1 \rangle} \cdot \mathbf{x}, \mathbf{p}^{\langle 2 \rangle} \cdot \mathbf{x}, \dots, \mathbf{p}^{\langle d-1 \rangle} \cdot \mathbf{x}\right).
\end{equation*}
The modified codeword of an information-symbol-local ML-LRC is then:
\begin{align*}
& \Big( \mathbf{x}, \mathbf{p}^{\langle 1 \rangle}_{\mathcal{I}_{1,1}} \cdot \mathbf{x}_{\mathcal{I}_{1,1}}, \dots, \mathbf{p}^{\langle 1 \rangle}_{\mathcal{I}_{1,\kappa_1}} \cdot  \mathbf{x}_{\mathcal{I}_{1,\kappa_1}}, \mathbf{p}^{\langle 1 \rangle}_{\mathcal{I}_{2,1}} \cdot \mathbf{x}_{\mathcal{I}_{2,1}}, \dots, \\
& \qquad \mathbf{p}^{\langle 1 \rangle}_{\mathcal{I}_{s,\kappa_{s}}} \cdot  \mathbf{x}_{\mathcal{I}_{s,\kappa_{s}}}, \mathbf{p}^{\langle 2 \rangle} \cdot \mathbf{x}, \dots, \mathbf{p}^{\langle d-1 \rangle} \cdot \mathbf{x} \Big).
\end{align*}
\end{construction}
Clearly, the obtained multiple-locality Pyramid code has minimum Hamming distance $d$ and length $k+ \sum_{i \in \inter{s}}\lceil k_i / \loc{i} \rceil + d-2$.
\begin{lemma}[Singleton-Optimal Construction]
The ML-Pyramid code as in Construction~\ref{constr_PyramidCode} has $n_i= k_i + \lceil k_i /\loc{i} \rceil =  \lceil k_i (\loc{i}+1) /\loc{i} \rceil $ symbols with locality $\loc{i}$ for all $i \in \inter{s}$. We get from the bound in Thm.~\ref{theo_SingletonDifferentLocalitiesSeveral}:
\begin{align}
 d & \leq n - k + 2 - \sum_{i \in \inter{s \minus 1}} \left \lceil \frac{n_i}{\loc{i}+1} \right \rceil - \left \lceil \frac{k \minusb \sum \limits_{i \in \inter{s \minus 1}} \loc{i}\left \lceil \frac{n_i}{\loc{i}+1} \right \rceil }{\loc{s}} \right \rceil \nonumber \\
& = n - k + 2 - \sum_{i \in \inter{s \minus 1}} \left \lceil \frac{k_i}{\loc{i}} \right \rceil - \left \lceil \frac{k \minusb \sum \limits_{i \in \inter{s \minus 1}} \loc{i}  \left \lceil \frac{k_i}{\loc{i}} \right \rceil }{\loc{s}} \right \rceil \label{eq_SingletonAdaptedPyramidCodes}.
\end{align}
With $k_{s-1} = k - \sum_{i \in \inter{s-1}} k_i$, we obtain from~\eqref{eq_SingletonAdaptedPyramidCodes}:
\begin{align*}
d & \leq n - k + 2 - \sum_{i \in \inter{s-1}} \left \lceil \frac{k_i}{\loc{i}} \right \rceil - \left \lceil \frac{k_{s-1} }{\loc{s}} \right \rceil, 
\end{align*}
which coincides with the minimum Hamming distance of the constructed ML-Pyramid code.
\end{lemma}

\section{Alphabet-Dependent Bound for ML-LRCs} \label{sec_Alphabet-DependentBound}
\subsection{Bound and Shortening}
For a given $(n,k,d)_q$ code $\CYC$ and a subset $\mathcal{I} \subseteq \inter{n}$ define as in \cite{cadambe_upper_2013}
\begin{equation} \label{eq_Entropy}
H(\mathcal{I}) \defeq \frac{\log |\{ \mathbf{x}_{\mathcal{I}} : \mathbf{x} \in \CYC  \}|}{\log q}.
\end{equation}
The bound given in Thm.~\ref{theo_CMBound} follows from the following two lemmas.
\begin{lemma}[{\cite[Lemma 1]{cadambe_upper_2013}}] \label{lem_CM1}
For a given $(n,k,d)_q$ $r$-local code $\CYC$, an integer $t \in \inter{\lceil k/r \rceil} $, a set $\mathcal{I} \subseteq \inter{n}$ with $|\mathcal{I}| =  t(r+1)$ and $H(\mathcal{I}) \leq tr$ as defined in~\eqref{eq_Entropy} exists (and is constructed explicitly).
\end{lemma}
\begin{lemma}[{\cite[Lemma 2]{cadambe_upper_2013}}] \label{lem_CM2}
For an $(n,k,d)_q$ code $\CYC$ with $\mathcal{I} \subseteq \inter{n}$ and $H(\mathcal{I}) \leq m$ an $(n-|\mathcal{I}|,k-m,d)_q$ code exists.
\end{lemma}
In the following lemma, we refine Lemma~\ref{lem_CM1} slightly and generalize it to an $((n_1, \loc{1}),(n_2, \loc{2}), \dots , (n_{s}, \loc{s}))$-local ML-LRC.
\begin{lemma} \label{lem_DifferentLocalities}
Let $\locset{i}$ with $|\locset{i}| = n_i $ for all $i \in \inter{s}$ be the locality sets of a given $(n,k,d)_q$ $((n_1, \loc{1}),(n_2, \loc{2}),\dots,(n_{s}, \loc{s}))$-local code $\CYC$ as in Definition~\ref{def_MultipleLRC}. Then, there exist $s$ sets $\mathcal{I}_i \subseteq  \locset{i}$ for all $i \in \inter{s}$ with
\begin{align*}
|\mathcal{I}_i|  =
\begin{cases}
t_i(\loc{i} +1) & \text{for } t_i \leq n_{i}/(\loc{i}+1), \; \text{if} \; (\loc{i}+1) \mid n_i\\
n_i & \text{for } t_i \leq \lceil n_i/(\loc{i}+1) \rceil, \; \text{if} \; (\loc{i}+1) \nmid n_i,
\end{cases}
\end{align*}
and $H(\mathcal{I}_i)  \leq t_i \loc{i}$ (for both cases) and for all $i \in \inter{s}$.
\end{lemma}
\begin{proof}
Similar to \cite[Proof of Lemma 1]{cadambe_upper_2013}, we construct the set $\mathcal{I}_{i}$ explicitly. Let $\mathcal{R}^{\langle j \rangle}_{i}$ denote the repair set of the $j$-th coordinate that belongs to $\locset{i}$. Clearly $|\mathcal{R}_{i}^{\langle j \rangle}| \leq \loc{i}$. Algorithm~\ref{algo_ConstructEntropySet} constructs the set $\mathcal{I}_{i}$.\\
\begin{center}
%\vspace{-.7cm}
\begin{algorithm}[htb]
\label{algo_ConstructEntropySet}
\SetAlgoVlined
\DontPrintSemicolon
\LinesNumbered
\SetKwInput{KwSub}{Subroutines}
\SetKwInput{KwIn}{Input}
\SetKwInput{KwOut}{Output}
\newcommand\mycommfont[1]{\footnotesize\ttfamily{#1}}
\SetCommentSty{mycommfont}
\BlankLine
\KwIn{Set $\locset{i}$ with $|\locset{i}| = n_i$ and locality $\loc{i}$.\\
\textcolor{white}{\textbf{Input:} }Integer $t_i \in \inter{\lceil n_i/(\loc{i}+1) \rceil}$.
}
\BlankLine 
\BlankLine 
Select $a_0$ arbitrarily from $\locset{i}$.\\
$\mathcal{S}^{\langle 0  \rangle} \leftarrow \emptyset$.\\
\For{$m  \in \inter{t_i-1}$}
{
Select $a_m$ in  $\locset{i} \setminus \left( \bigcup_{l=0}^{m-1} \{a_m \cup {\mathcal{R}_{i}^{\langle a_l \rangle}} \cup \mathcal{S}^{\langle l \rangle} \right)$.\\
$\mathcal{I}^{\langle m \rangle}_i \leftarrow  \bigcup_{l=0}^{m-1} \{a_l\} \cup {\mathcal{R}_{i}^{\langle a_l \rangle}} \cup \mathcal{S}^{\langle l \rangle}$.\\ 
Select $\mathcal{S}^{\langle m \rangle} \subseteq \locset{i} \setminus \left( \{ a_m \} \cup \mathcal{R}_{i}^{\langle a_m \rangle} \cup \mathcal{I}^{\langle m \rangle}_i \right)$ with\\
\hspace{-.15cm}$|\mathcal{S}^{\langle m \rangle}| = \min((m \text{+}1)(\loc{i} \text{+}1),n_i) \minusb | \{ a_m \} \cup \mathcal{R}_{i}^{\langle a_m \rangle} \cup \mathcal{I}^{\langle m \rangle}_i |$ 
}
$\mathcal{I}_i \leftarrow  \bigcup_{m=0}^{t_i-1} \{a_m \} \cup {\mathcal{R}_{i}^{\langle a_m \rangle}} \cup \mathcal{S}^{\langle m \rangle}$. 
\BlankLine
\BlankLine
\KwOut{Return the set $\mathcal{I}_i$.}
\caption{Construction of $\mathcal{I}_{i} \subseteq \locset{i}$.}
\end{algorithm}
%\vspace{-1cm}
\end{center}
Algorithm~\ref{algo_ConstructEntropySet} differs from the algorithm used in \cite[Proof of Lemma 1]{cadambe_upper_2013} only in Line 7, which ensures that the constructed set cannot have a cardinality greater than $n_i$. In the proof of \cite[Lemma 1]{cadambe_upper_2013}, it is shown that $H(\mathcal{I}_i) = H(\mathcal{I}_i\setminus \{a_0,a_1,\dots,a_{t_i-1} \})$.
Clearly, for:\\
Case a): If $t_i \leq \lfloor n_i/(\loc{i}+1) \rfloor$, then $|\mathcal{I}_{i}| = t_i(\loc{i}+1)$ and $H(\mathcal{I}_i) \leq t_i \loc{i}$. \\
Case b): If $t_i = \lceil n_i/(\loc{i}+1) \rceil$ and $(\loc{i}+1) \nmid n_i$, then $|\mathcal{I}_{i}| = n_i$ and 
\begin{align}
H(\mathcal{I}_i) & \leq n_i - \lceil n_i/(\loc{i}+1) \rceil \nonumber \\
 & = \left \lceil \frac{\loc{i}n_i}{\loc{i}+1} \right \rceil. \label{eq_TechnicalSet}
\end{align}
\end{proof}
The distinction between Cases a) and b) in Lemma~\ref{lem_DifferentLocalities} is not relevant for the bound of an $r$-local code, but becomes necessary if we want to bound the dimension of an $((n_1, \loc{1}),(n_2, \loc{2}),\dots,(n_{s}, \loc{s}))$-local ML-LRC.
\begin{lemma} \label{lem_SeveralEntropySets}
For an $(n,k,d)_q$ code $\CYC$ with $s$ pairwise disjoint sets $\mathcal{I}_i \subseteq \inter{n}$ with $H(\mathcal{I}_i) \leq m_i$ for all $i \in \inter{s}$, there exists an $(n-\sum_{i \in \inter{s}}|\mathcal{I}_i|,k-\sum_{i \in \inter{s}}m_i,d)$ code.
\end{lemma}
\begin{proof}
Apply Lemma~\ref{lem_CM2} consecutively $s$ times and the statement follows.
\end{proof}
From Lemma~\ref{lem_DifferentLocalities} and Lemma~\ref{lem_SeveralEntropySets} the following theorem follows.
\begin{theorem}[Alphabet-Dependent Bound for ML-LRCs with Two Different Localities] \label{thm_AlphabetDiffLocalities}
Let $\CYC$ be an $(n,k,d)_q$ $((n_1, \loc{1}), (n_2, \loc{2}))$-local code. 
Define
\begin{equation} \label{eq_CMboundTwoLocalities1}
\begin{split}
& k_{t_1,t_2}  \defeq \\
& t_1 \loc{1} \text{+} t_2 \loc{2} \text{+} k_{\text{opt}}^{(q)}\left(n \minusb \min(n_1,t_1( \loc{1} \text{+}1)) \minusb \min(n_2,t_2(\loc{2}\text{+}1)), d \right).
\end{split}
\end{equation}
The dimension of $\CYC$ is bounded from above by
\begin{equation} \label{eq_CMboundTwoLocalities2}
k \leq \min_{\substack{t_1 \in \inter{t_1^{\star}},\\ t_2 \in \inter{t_2^{\star}}}} k_{t_1,t_2}, 
\end{equation}
where 
\begin{align}
t_1^{\star} & \defeq \left \lceil \frac{n_1}{\loc{1}+1} \right \rceil, \label{eq_CondRunning1}\\
t_2^{\star} & \defeq \left \lfloor \frac{k - 1 - t_1 \loc{1}}{\loc{2}} \right \rfloor, \quad \forall t_1 \in \inter{t_1^{\star}}.\label{eq_CondRunning2}
\end{align}
\end{theorem}
\begin{proof}
The expressions as in~\eqref{eq_CMboundTwoLocalities1} and \eqref{eq_CMboundTwoLocalities2} follow from Lemma~\ref{lem_DifferentLocalities} and Lemma~\ref{lem_SeveralEntropySets} for $s=2$.

The value of $t_1$ can be bounded by $\min (\lceil n_1/(\loc{1}+1) \rceil, \lfloor (k-1)/\loc{1} \rfloor)$ similar as in the case of an $\loc{}$-local LRC. We assume that $\loc{1} \lceil n_1 / (\loc{1}+1) \rceil < k-1$ and therefore we obtain~\eqref{eq_CondRunning1}. 
The maximal value for $t_2$ follows from the fact that for  $ t_2 > t_2^*$ the expression in~\eqref{eq_CMboundTwoLocalities2} is at least $k$.
Clearly, $t_2 \leq \lceil n_2/(\loc{2}+1) \rceil$. It is well-known that the rate of an $\loc{}$-local code is at most $\loc{}/(\loc{}+1)$ and due to the fact that an $((n_1, \loc{1}), (n_2, \loc{2}))$-local code is also an $\loc{2}$-local code, we have:
\begin{equation*}
\left \lfloor \frac{k-1-\loc{1}t_1}{\loc{2}} \right \rfloor \leq \left \lceil \frac{n_2}{\loc{2}+1} \right \rceil,
\end{equation*}
and therefore the maximal $t_2^{\star}$ as in~\eqref{eq_CondRunning2} follows.
\end{proof}
\begin{corollary}[Singleton vs. Alphabet-Dependent Bound] \label{cor_SingletonVsAlphabetTwo}
The Singleton bound as in Thm.~\ref{theo_SingletonDifferentLocalities} for an $\LIN{n}{k}{d}{q}$ $((n_1, \loc{1}), (n_2, \loc{2}))$-local code is weaker than the bound given in Thm.~\ref{thm_AlphabetDiffLocalities}.
\end{corollary}
\begin{proof}
We first bound the dimension $k_{\text{opt}}^{(q)}$ by the locality-unaware Singleton bound.
\begin{align} 
k  & \leq \min_{t_1, t_2} \Big( t_1 \loc{1} + t_2 \loc{2} +  k_{\text{opt}}^{(q)}(n-t_1(\loc{1}+1)-t_2(\loc{2}+1),d) \Big) \nonumber \\
& \leq t_1 \loc{1} + t_2 \loc{2} + n - t_1(\loc{1}+1)- t_2(\loc{2}+1)-d+1 \nonumber \\
& = n-d+1-t_1-t_2 \label{eq_FirstCMBound}.
\end{align}
We assume that $ d > n - k + 2 - \lceil n_1/(\loc{1}+1) \rceil - \lceil (k- \loc{1}\lceil n_1/(\loc{1}+1) \rceil) /\loc{2} \rceil$ and inserted in~\eqref{eq_FirstCMBound} leads to:
\begin{align} 
k & < n - n + k - 2 +  \left \lceil \frac{n_1}{\loc{1}+1} \right \rceil + \left \lceil \frac{k- \loc{1} \left \lceil \frac{n_1}{\loc{1}+1} \right \rceil}{\loc{2}} \right \rceil  \nonumber \\
 & \qquad + 1 - t_1 - t_2 \nonumber  \\
  & = k - 1 + \left \lceil \frac{n_1}{\loc{1}+1} \right \rceil + \left \lceil \frac{k- \loc{1} \left \lceil \frac{n_1}{\loc{1}+1} \right \rceil}{\loc{2}} \right \rceil - t_1 - t_2. \label{eq_ContraLastLine}
\end{align} 
For the maximal values of $t_1^*$ and $t_2^*$ as in~\eqref{eq_CondRunning1} and \eqref{eq_CondRunning2}, \eqref{eq_ContraLastLine} gives a contradiction.
\end{proof}
Note that in the classical case of an $\loc{}$-local LRC the Cadambe--Mazumdar bound (Thm.~\ref{theo_CMBound}) gives also a contradiction (if $d > n-k+2-\lceil k/r \rceil$) for the maximal value of $t = \lceil k/r \rceil$. 

The following theorem generalizes Thm.~\ref{thm_AlphabetDiffLocalities} to the case of $s>2$ different localities.
\begin{theorem}[Alphabet-Dependent Bound for ML-LRCs] \label{thm_AlphabetDiffLocalitiesSeveral}
Let $\CYC$ be an $(n,k,d)_q$ $((n_1, \loc{1}), (n_2, \loc{2}),\dots,(n_{s}, \loc{s}))$-local code as in Definition~\ref{def_MultipleLRC}. Furthermore, assume $ \sum_{i \in \inter{s}} \loc{i} \lceil n_i/(\loc{i}+1) \rceil < k-1 $. Define
\begin{align}
t_i^{\star} & \defeq \left \lceil \frac{n_i}{\loc{i}+1} \right \rceil, \quad \forall i \in \inter{s-1}, \label{eq_CondRunning1Several}\\
t_{s}^{\star} & \defeq \left \lfloor \frac{k-1 - \sum \limits_{i \in \inter{s-1}} t_i\loc{i} }{\loc{s}} \right \rfloor, \quad \forall t_i \in \inter{t_i^{\star}}, \label{eq_CondRunning2Several}
\end{align}
and let 
\begin{equation} \label{eq_CMboundSeveralLocalities}
\begin{split}
& k_{t_1,t_2,\dots, t_{s}} \defeq \\
& \sum_{i \in \inter{s}} t_i \loc{i} + k_{\text{opt}}^{(q)}\left(n-\sum_{i \in \inter{s}} \min(n_i, t_i(\loc{i}+1)), d \right).
\end{split}
\end{equation}
Then the dimension of $\CYC$ is upper bounded by
\begin{equation}
k \leq \min_{\substack{t_i \in \inter{t_i^{\star}} \\ \forall i \in \inter{s}}} k_{t_1,t_2, \dots, t_{s}}.
\end{equation}
\end{theorem}
\begin{proof}
See appendix.
\end{proof}
The following corollary generalizes Corollary~\ref{cor_SingletonVsAlphabetTwo} to the case of $s>2$ different localities.
\begin{corollary}[Singleton vs. Alphabet-Dependent Bound] \label{cor_SingletonVSAlphabetSeveral}
The Singleton bound as in Thm.~\ref{theo_SingletonDifferentLocalitiesSeveral} for an $\LIN{n}{k}{d}{q}$ 
$((n_1, \loc{1}), (n_2, \loc{2}),$ ... $,(n_{s}, \loc{s}))$-local code is weaker than the bound given in Thm.~\ref{thm_AlphabetDiffLocalitiesSeveral}.
\end{corollary}
\begin{proof}
See appendix.
\end{proof}
If there exist $s$ parameters $t_i \in \inter{t_i^{\star}}, \forall i \in \inter{s}$ such that the dimension $k$ of an $\LIN{n}{k}{d}{q}$ $((n_1, \loc{1}), (n_2, \loc{2}),\dots , (n_{s}, \loc{s}))$-local ML-LRC equals $k_{t_1,t_2,\dots,t_s}$ as in~\eqref{eq_CMboundSeveralLocalities}, then we call the ML-LRC \textit{\CMoptimal}.

\subsection{Shortening and Construction}
Similar to Lemma~\ref{lem_shorteningLRC}, we now prove that shortening an $\loc{}$-local LRC, which is optimal with respect to the bound given in Thm.~\ref{theo_CMBound}, by one coordinate gives an $((\loc{}, \loc{}-1), (n-1-\loc, \loc{}))$-local \CMoptimal{} ML-LRC. We give an explicit construction (in Algorithm~\ref{algo_ConstructTwoDiffLocalitiesLRCAlphabetDependent}) of a ML-LRC with two different localities based on a given an \CMoptimal{} LRC.
\begin{lemma}[Shortening an $\loc{}$-local Code] \label{lem_ShorteningLRCCM}
Let an $\LIN{n}{k}{d}{q}$ $\loc{}$-local \CMoptimal{} code be given. Then, the shortened $\LIN{n-1}{k-1}{d}{q}$ $((\loc{}, \loc{}-1), (n-1-\loc, \loc{}))$-local code is \CMoptimal{}.
\end{lemma}
\begin{proof}
Clearly, the shortening affects one repair-set and its cardinality is $\loc{}$ and the locality of the contained code symbols is $\loc{}-1$.
Let $t'$ be the index of the minimal value $k_{t'}$ of \eqref{eq_CMBoundb} such that the dimension of the given code is 
\begin{equation*}
k = k_{t'} = t'\loc{} + k_{\text{opt}}^{(q)}(n-t'(\loc{}+1), d).
\end{equation*}
We obtain from Thm.~\ref{thm_AlphabetDiffLocalities} with $n_1=\loc{}$, $\loc{1} = \loc{}-1$, $\loc{2}=\loc{}$, $t_1=1$ and $t_2=t'-1$ for the shortened code of length $n-1$ and dimension $k-1$:
\begin{align*}
k_{1,t'-1} & = \loc{} \minusb 1 + (t' \minusb 1) \loc{} + k_{\text{opt}}^{(q)}(n-1-\loc{}-(t' \minusb 1)(\loc{}+1), d) \\
& = t' \loc{} - 1 + k_{\text{opt}}^{(q)}(n-t'(\loc{}+1), d),
\end{align*}
which equals $k_{t'} -1$. Therefore, the dimension of the shortened code meets the bound of Thm.~\ref{thm_AlphabetDiffLocalities} and is \CMoptimal{}.
\end{proof}
Algorithm~\ref{algo_ConstructTwoDiffLocalitiesLRCAlphabetDependent} returns an ML-LRC with two localities that is optimal with respect to the bound given in Thm.~\ref{thm_AlphabetDiffLocalities}.
\begin{center}
\vspace{-.3cm}
\begin{algorithm}[h]
\label{algo_ConstructTwoDiffLocalitiesLRCAlphabetDependent}
%\IncMargin{1em}
\SetAlgoVlined
\DontPrintSemicolon
\LinesNumbered
%\Indm  
\SetKwInput{KwSub}{Subroutines}
\SetKwInput{KwIn}{Input}
\SetKwInput{KwOut}{Output}
\newcommand\mycommfont[1]{\footnotesize\ttfamily{#1}}
\SetCommentSty{mycommfont}
\BlankLine
\KwIn{\\Parity-check matrix $\mathbf{H}$ of an $\LIN{n}{k}{d}{q}$ $\loc{2}$-local \CMoptimal{} code with $\rho = \lceil n/(\loc{2}+1) \rceil$ repair sets $\mathcal{R}_1, \mathcal{R}_2, \dots, \mathcal{R}_{\rho}$.\\
$\alpha \in \inter{\rho}$.
}
\BlankLine 
%\Indp
\BlankLine 
\For{$i \in \inter{\alpha}$}
{
Delete $\loc{2}-\loc{1}$ columns of $\mathbf{H}$ that are contained in $\mathcal{R}_i$.
}
\BlankLine
\BlankLine
\KwOut{Parity-check matrix $\mathbf{H}$ of an $\LIN{n-\alpha(\loc{2}-\loc{1})}{k-\alpha(\loc{2}-\loc{1})}{d}{q}$ $ ((\alpha(\loc{1}+1), \loc{1}), (n - \alpha(\loc{2}+1), \loc{2}))$-local code.}
\caption{Alphabet-Optimal ML-LRC over small field size.}
\end{algorithm}
%\vspace{-1cm}
\end{center}
For a sufficiently large field-size a Singleton-optimal LRC can be constructed for a wide choice of parameters, but an \CMoptimal{} LRC for a (aimed) smaller field-size are not known for all parameters. Therefore, the parameters in Algorithm~\ref{algo_ConstructTwoDiffLocalitiesLRCAlphabetDependent} are expressed in terms of the given parameters $n,k,d$ of the given \CMoptimal{} $\loc{}$-local code.
\begin{theorem}[Alphabet-Optimal ML-LRC] \label{thm_ConstructionOptimalMLLRC}
Given an $\LIN{n}{k}{d}{q}$ $\loc{2}$-local \CMoptimal{} LRC and a integer $\alpha \in \inter{\lceil n / (\loc{2}+1) \rceil}$, Algorithm~\ref{algo_ConstructTwoDiffLocalitiesLRCAlphabetDependent} returns a parity-check matrix of an  $\LIN{n-\alpha(\loc{2}-\loc{1})}{k-\alpha(\loc{2}-\loc{1})}{d}{q}$ $ ((\alpha(\loc{1}+1), \loc{1}), (n - \alpha(\loc{2}+1), \loc{2}))$-local ML-LRC that is \CMoptimal{}.
\end{theorem}
\begin{proof}
For the given $\LIN{n}{k}{d}{q}$ $\loc{2}$-local \CMoptimal{} code, there exists an integer $t'$ such that
\begin{equation}
k = k_t' = t' \loc{2} +  k_{\text{opt}}^{(q)}(n-t'(\loc{2}+1),d).
\end{equation}
The shortened code is an $\LIN{n-\alpha(\loc{2}-\loc{1})}{k-\alpha(\loc{2}-\loc{1})}{d}{q}$ $ ((\alpha(\loc{1}+1), \loc{1}), (n - \alpha(\loc{2}+1), \loc{2}))$-local code. With $t_1=\alpha$ and $t_2=t'-\alpha$, we obtain with the new alphabet-dependent bound of Thm.~\ref{thm_AlphabetDiffLocalities} that the dimension of the shortened code is upper bounded by
\begin{align}
k_{\alpha,t' \minus \alpha} & = \alpha \loc{1} + (t' \minusb \alpha) \loc{2} + \nonumber \\
& \quad k_{\text{opt}}^{(q)}(n \minusb \alpha(\loc{2} \minusb \loc{1}) \minusb \alpha(\loc{1}+1) \minusb (t' \minusb \alpha)(\loc{2}+1),d), \nonumber
\end{align} 
and this leads to:
\begin{align} 
k_{\alpha,t' \minus \alpha} & = t' \loc{2} - \alpha(\loc{2}-\loc{1}) + k_{\text{opt}}^{(q)}(n-t'(\loc{2}+1) , d), \nonumber \\
& = k_{t'}-\alpha(\loc{2}-\loc{1}). \label{eq_DimensionMLLRC}
\end{align} 
The expression in~\eqref{eq_DimensionMLLRC} equals the dimension of the shortened $((\alpha(\loc{1}+1), \loc{1}), (n - \alpha(\loc{2}+1), \loc{2}))$-local ML-LRC and therefore it is optimal with respect to the new alphabet-dependent bound of Thm.~\ref{thm_AlphabetDiffLocalities}.
\end{proof} 
Similar to Algorithm~\ref{algo_ConstructTwoDiffLocalitiesLRC}, Algorithm~\ref{algo_ConstructTwoDiffLocalitiesLRCAlphabetDependent} can be easily be extended to $s \geq 2$ localities.

\section{Alphabet-Optimal Binary LRCs and ML-LRCs} \label{sec_OptimalBinaryGCC}
This section provides a family of binary $\loc{}$-local LRCs. Their construction is based on generalized (or multilevel) code concatenation introduced for linear components by Blokh and Zyablov~\cite{blokh_coding_1974} and described for not necessarily linear components by Zinoviev~\cite{zinoviev_generalized_1976}. Multilevel code concatenation generalizes classical concatenation introduced by Forney~\cite{forney_concatenated_1966}. Furthermore, we give an example of a shortened $\loc{}$-local code and show that it is optimal with respect to our new alphabet-dependent bound for ML-LRCs from Thm.~\ref{thm_AlphabetDiffLocalities}.

Let us recall some definitions and facts before we define a generalized concatenated code. The Kronecker product of two matrices $\mathbf{A} = (a_{ij})_{i \in \inter{m}}^{j \in \inter{m}}$ and $\mathbf{B}$ is defined as $\mathbf{A} \otimes \mathbf{B} = (a_{ij}\mathbf{B})_{i \in \inter{m}}^{j \in \inter{m}}$. For two linear codes $\mathcal{A}$ and $\mathcal{B}$ of same length, define the direct sum code as $\mathcal{A} \oplus \mathcal{B} \defeq \lbrace \mathbf{a} \oplus \mathbf{b} | \mathbf{a} \in \mathcal{A}, \mathbf{b} \in \mathcal{B} \rbrace$. For a linear subcode $\mathcal{B}^{\langle 2 \rangle}$ of $\mathcal{B}^{\langle 1 \rangle}$ with generator matrix $\mathbf{G}_{\mathcal{B}^{\langle 2 \rangle }}$, let $\mathcal{B}^{\langle 2 \rangle} \setminus \mathcal{B}^{\langle 1 \rangle}$ be such that $\mathcal{B}^{\langle 1 \rangle} = \mathcal{B}^{\langle 2 \rangle} \oplus \mathcal{B}^{\langle 1 \rangle}\setminus \mathcal{B}^{\langle 2 \rangle}$. It is well-known (see, e.g.,~\cite{bossert_results_1999}) that there exists a generator matrix of $\mathcal{B}^{\langle 1 \rangle}$ that can be written as 
\begin{equation*}
\mathbf{G}_{\mathcal{B}^{\langle 1 \rangle }} = 
\begin{pmatrix}
\mathbf{G}_{\mathcal{B}^{\langle 2 \rangle }}\\
\mathbf{G}_{\mathcal{B}^{\langle 2 \rangle } \setminus \mathcal{B}^{\langle 1 \rangle }}
\end{pmatrix}.
\end{equation*}
Now, we can define generalized code concatenation.
\begin{definition}[Generalized Concatenated Code (GCC)] \label{def_GCC}
Let $\GCCOUT{i}$ be an outer $\LIN{\GCCOUTn{}}{\GCCOUTk{i}}{\GCCOUTd{i}}{q^{\GCCOUTextq{i}}}$ code with a generator matrix $\genmat[\GCCOUT{i}] \in \F{q^{\GCCOUTextq{i}}}^{\GCCOUTk{i} \times \GCCOUTn{}}$ and let $\GCCIN{i}$ be an $\LIN{\GCCINn{}}{\GCCINk{i}}{\GCCINd{i}}{q}$ inner code for all $i \in \inter{s}$. 
Furthermore, let
\begin{equation*}
  \GCCIN{1} \supset \GCCIN{2} \supset \dots \supset \GCCIN{s}, 
\end{equation*}
and let $\genmat[\GCCIN{i}\setminus \GCCIN{i+1}] \in \F{q}^{(\GCCINk{i}-\GCCINk{i+1}) \times \GCCINn{}}$ denote a generator matrix of the code $\GCCIN{i}\setminus \GCCIN{i+1}$ for all $i \in \inter{s-1}$.
Let $\GCCINk{i}-\GCCINk{i+1} = \lambda_i \GCCOUTextq{i}$ for all $i \in \inter{s-1}$ and let $\GCCINk{s} = \lambda_{s} \GCCOUTextq{s}$.

The matrix $\genmat[\GCCOUT{i}][b] \in \F{q}^{\GCCOUTk{i} \GCCOUTextq{i} \times \GCCOUTn{}}$ is the given $\GCCOUTk{i} \times \GCCOUTn{}$ generator matrix $\genmat[\GCCOUT{i}]$ represented over $\F{q}$.
The code with generator matrix 
\begin{equation} \label{eq_GenMatGCC}
\genmat = 
\begin{pmatrix}
\genmat[\GCCOUT{1}][b] \otimes \genmat[\GCCIN{1} \setminus \GCCIN{2}] \\
\genmat[\GCCOUT{2}][b] \otimes \genmat[\GCCIN{2} \setminus \GCCIN{3}] \\
\vdots \\
\genmat[\GCCOUT{s-1}][b] \otimes \genmat[\GCCIN{s-1} \setminus \GCCIN{s}] \\
\genmat[\GCCOUT{s}][b] \otimes \genmat[\GCCIN{s}]
\end{pmatrix}
\end{equation}
is an $\LIN{\GCCn = \GCCOUTn{} \GCCINn{} }{\GCCk = \sum_{i \in \inter{s}} \GCCOUTk{i} \lambda_i \GCCOUTextq{i}}{\GCCd}{q}$ $s$-level concatenated code $\CYC$.
\end{definition}
It is well-known that the minimum Hamming distance of a generalized concatenated code is $\GCCd \geq \min_{i \in \inter{s}} \GCCOUTd{i} \GCCINd{i}$ (see, e.g., \cite{blokh_coding_1974, zinoviev_generalized_1976}). Every column of $\genmat$ of \eqref{eq_GenMatGCC} belongs to $\GCCIN{1}$ (assuming that the inner codewords are represented as rows in the codeword matrix). The following construction is based on this fact and generalized concatenated code inherits the locality property from $\GCCIN{1}$.
\begin{construction}[$2$-Level $\loc{}$-local GCC] \label{constr_GCC1}
Let $\loc{}>1$ be the given locality parameter and let a integer be $j < 2^{\loc{}-1}-\loc{}+1$. Let the two outer codes and the two inner codes of a $2$-level concatenated code as in Definition~\ref{def_GCC} be a:\\[1ex]
\begin{tabular}{ll}
$\GCCOUT{1}$: & $\LIN{2^{\loc{}-1} +1 -j}{2^{\loc{}-1}-\loc{} + 1 -j}{\loc{} +1}{2^{\loc{} - 1}}$\\
 & MDS code,\\
$\GCCOUT{2}$: & $\LIN{2^{\loc{}-1} +1 - j}{2^{\loc{}-1}-j}{2}{2}$\\
 &  single-parity check code,\\
$\GCCIN{1}$: & $\LIN{\loc{}+1}{\loc{}}{2}{2}$ single-parity check code,\\
$\GCCIN{2}$: & $\LIN{\loc{}+1}{1}{\loc{}+1}{2}$ repetition code.
\end{tabular}\\[1ex]
Then the $2$-level concatenated code is an 
\begin{equation*}
\LIN{(\loc{}+1)(2^{\loc{}-1}+1-j)}{\loc{}(2^{\loc{}-1} -r +2-j) - 1}{2(\loc{}+1)}{2}
\end{equation*}
$\loc{}$-local code.
\end{construction}
The following theorem shows the optimality of Construction~\ref{constr_GCC1} with respect to the bound given in Thm.~\ref{theo_CMBound} for $\loc{}=2,3$.
\begin{theorem}[Alphabet-Optimal Binary LRC] \label{theo_OptimalityConstrGCC2}
For $r=2,3$ Construction~\ref{constr_GCC1} gives \CMoptimal{} LRCs.
\end{theorem}
\begin{proof}
It is sufficient to show it for the case of $j=0$. The general case $j > 0$ follows directly, because the length of the code is reduced by $j(\loc{}+1)$ and the dimension by $j \loc{}$ and Thm.~\ref{theo_CMBound} gives the same bound for $t=t'-j$ as $t'$ for the code with $j=0$.

The upper bound on the dimension of an $\loc{}$-local code from Thm.~\ref{theo_CMBound} gives for a binary code as in Construction~\ref{constr_GCC1} for $t=2^{r - 1} - r+1$:
\begin{align} 
& k_{2^{r - 1} - r+1}  \nonumber \\
& = r (2^{r \minus 1} -r+1) \nonumber \\
& \: + k_{opt}^{(2)}\left((r+1)(2^{r \minus 1}+1) \minusb (2^{r \minus 1}\minusb r+1)(r+1), 2(r+1)\right) \nonumber \\
& = r (2^{r \minus 1} - r+1) + k_{opt}^{(2)}\big(r(r+1), 2(r+1)\big).  \label{eq_xxx}
\end{align}
For $r=2,3$, the maximal dimension of a binary code of length $r(r+1)$ and minimum Hamming distance $2(r+1)$ is $k_{opt}^{(2)}\big(r(r+1), 2(r+1)\big) = r-1$ and therefore $k_{2^{r - 1} - r+1}$ equals the dimension of the 2-level generalized concatenated code of Construction~\ref{constr_GCC1}.
\end{proof}
We now give an example of a binary \CMoptimal{} ML-LRC obtained via shortening an \CMoptimal{} LRC as in Thm.~\ref{theo_OptimalityConstrGCC2}.
\begin{example}[Alphabet-Optimal Binary ML-LRC]
Let $r=3$ in Construction~\ref{constr_GCC1}, then the two outer and inner codes of Construction~\ref{constr_GCC1} are:\\[1ex]
\begin{tabular}{ll}
$\GCCOUT{1}$: & $\LIN{5}{2}{4}{2^2}$ MDS code,\\
$\GCCOUT{2}$: & $\LIN{5}{4}{2}{2}$ single-parity check code,\\
$\GCCIN{1}$: & $\LIN{4}{3}{2}{2}$ single-parity check code, \\
$\GCCIN{2}$: & $\LIN{4}{1}{4}{2}$ repetition code.
\end{tabular}\\[1ex]
The $2$-level concatenated code $\CYC$ is a $\LIN{20}{8}{8}{2}$ $3$-local \CMoptimal{} code.

The bound as in Thm.~\ref{theo_CMBound} gives the following upper bound on the dimension of such a binary code:
\begin{align*}
k_2 & = 2 \cdot 3 + k_{\text{opt}}^{(2)}(20-2(3+1), 8) \\
   & = 6  + k_{\text{opt}}^{(2)}(12, 8),
\end{align*}
where the maximal dimension of a binary code of length 12 and with minimum Hamming distance 8 is $k_{\text{opt}}^{(2)}(12, 8) =2$ and therefore $k_2 = 8$, which is the dimension of our constructed code $\CYC$.
Now we shortened the $\LIN{20}{8}{8}{2}$ $3$-local \CMoptimal{} code by one position and obtain an $\LIN{19}{7}{8}{2}$ $((3,2),(16,3))$-local code and we obtain from Thm.~\ref{thm_AlphabetDiffLocalities} that:
\begin{align*}
k_{1,1} & = 1 \cdot 2 + 1 \cdot 3 + k_{\text{opt}}^{(2)}(19-3-4, 8) \\
 & = 5 + k_{\text{opt}}^{(2)}(12, 8) = 7,
\end{align*}
which equals the dimension of the shortened $((3,2),(16,3))$-local ML-LRC.
\end{example}

\section{Conclusion and Outlook} \label{sec_Concl}
We introduced the class of multiple-locality LRCs. The Singleton-like upper bound on the minimum Hamming distance and the alphabet-dependent bound of Cadambe--Mazumdar on the dimension of an LRC were generalized to the case of ML-LRCs. Furthermore, we gave constructions of optimal codes with respect to the two new bounds based on shortening a Singleton respectively \CMoptimal{} LRC. An adapted Pyramid code construction for the case of different information-symbol locality was outlined. Moreover, we gave a class of binary $\loc{}$-local LRCs based on generalized code concatenation and showed that they are \CMoptimal{} for $\loc{}=2,3$ and also provide optimal ML-LRCs over the binary alphabet.

Future work are (direct) constructions of ML-LRC without shortening and the adaption of existing bounds and constructions for LRC where the parameter for every code symbol is also variable (e.g., availability).

\section*{Acknowledgment}
The authors thank I. Tamo for fruitful discussions (especially for Thm.~\ref{theo_SingletonDifferentLocalities}).

\printbibliography

\appendix

\begin{proof}[Proof of Thm.~\ref{theo_SingletonDifferentLocalitiesSeveral}]
Let $\mathbf{G}$ be a $k \times n$ generator matrix of the $\LIN{n}{k}{d}{q}$ $((n_1, \loc{1}),(n_2, \loc{2}),\dots,(n_{s}, \loc{s}))$-local code. We assume that $ \sum_{i \in \inter{s-1}} \loc{i} \lceil n_i/(\loc{i}+1) \rceil < k-1 $. %(see Note~\ref{note_AssumptionOnLocalitySeveral}).
\begin{enumerate}[leftmargin=1cm]
\item[i)] Choose $\kappa_i \defeq \lceil n_i/(\loc{i}+1) \rceil$ columns of $\mathbf{G}$ indexed by $j_1^{\langle i\rangle},j_2^{\langle i\rangle}, \dots,j_{\kappa_i}^{\langle i\rangle}$, where $j_{\iota}^{\langle i\rangle} \in \locset{i}$ for all $i \in \inter{s-1}$. Each column is a linear combination of at most $\loc{i}$ other columns. 
\item[s)] Let $\mathcal{I}^{\langle 1 \rangle}, \mathcal{I}^{\langle 2 \rangle}, \dots, \mathcal{I}^{\langle s-1 \rangle}$ be the $s-1$ sets of indexes of all repair columns indexed by $j_1^{\langle i\rangle},j_2^{\langle i\rangle}, \dots,j_{\kappa_i}^{\langle i\rangle}$ but without the indexes itself. We have $| \cup_{i \in \inter{s-1}}\mathcal{I}^{\langle i \rangle}| \leq \sum_{i \in \inter{s-1}}\loc{i} \cdot \kappa_i < k-1$.
\item[s+1)] Choose
\begin{equation*} 
\kappa_{s} \defeq \left \lfloor \frac{k-1- \sum \limits_{i \in \inter{s-1}}\loc{i}\left \lceil \frac{n_i}{\loc{i}+1} \right \rceil}{\loc{s}} \right \rfloor
\end{equation*} 
columns of $\mathbf{G}$ indexed by $j_1^{\langle s \rangle},j_2^{\langle s \rangle}, \dots,j_{\kappa_s}^{\langle s \rangle}$, where $j_{i}^{\langle s \rangle} \in \locset{s}$. Each column is a linear combination of at most $\loc{s}$ other columns. 
Let $\mathcal{I}^{\langle s \rangle} \subset \locset{s}$ be the set of indexes of all repair columns indexed by $j_1^{\langle s \rangle},j_2^{\langle s \rangle}, \dots,j_{\kappa_s}^{\langle s \rangle}$ but without the indexes themselves.
\item[s+2)] Let $ \mathcal{I} \defeq \cup_{i \in \inter{s}} \mathcal{I}^{\langle i \rangle}$. Then $| \mathcal{I}| < k$ and we have $\rank(\mathbf{G}_{\mathcal{I}}) < k$.
\item[s+3)] Enlarge $\mathcal{I}$ to the set $\mathcal{I}'$, s.t. $\rank(\mathbf{G}_{\mathcal{I}'}) = k-1$, but without using $\left\{ j_1^{\langle i \rangle},j_2^{\langle i \rangle},\dots,j_{\kappa_i}^{\langle i \rangle} \right\}_{i \in \inter{s}}$.
\item[s+4)] Define 
\begin{equation}
\begin{split}
\mathcal{U} & \defeq \left\{ \mathcal{I}' \cup \left\{ j_1^{\langle 1 \rangle},j_2^{\langle 1 \rangle},\dots,j_{\kappa_s}^{\langle s \rangle} \right\} \right\}, 
\end{split}
\end{equation}
where $| \mathcal{U}| \geq k-1 + \sum_{i \in \inter{s}} \kappa_i $.
\end{enumerate}
We have $\rank(\mathbf{G}_{\mathcal{U}}) \leq k-1$. From Lemma~\ref{lem_RankGeneratorMatrix}, we know that $|\mathcal{U}|$ can be upper bounded, i.e.:
\begin{align*}
& k -1 + \sum_{i \in \inter{s}} \kappa_i \leq n-d, 
\end{align*}
and this implies for the minimum Hamming distance of an $((n_1, \loc{1}),(n_2, \loc{2}),\dots,(n_{s}, \loc{s}))$-local code
\begin{align*}
d & \leq n \text{-} k \text{+} 1 - \sum_{i \in \inter{s \minus 1}} \left \lceil \frac{n_i}{\loc{i}+1} \right \rceil \text{-} \left \lfloor \frac{k \text{-} 1 - \sum \limits_{i \in \inter{s \minus 1}} \loc{i}\left \lceil \frac{n_i}{\loc{i}+1} \right \rceil }{\loc{s}} \right \rfloor \\
& = n \text{-} k \text{+} 2 - \sum_{i \in \inter{s \minus 1}} \left \lceil \frac{n_i}{\loc{i}+1} \right \rceil \text{-} \left \lceil \frac{k - \sum \limits_{i \in \inter{s \minus 1}} \loc{i}\left \lceil \frac{n_i}{\loc{i}+1} \right \rceil }{\loc{s}} \right \rceil.
\end{align*}
In case the assumption does not hold in Step $j)$, i.e., $\sum_{i \in \inter{j}} \loc{i} \lceil n_i/(\loc{i}+1) \rceil \geq k-1$, the bound becomes the bound for an $((n_1, \loc{1}),(n_2, \loc{2}), \dots , ( \sum_{i \in \inter{j,s}} n_{i}, \loc{j}))$-local code. Clearly, an $((n_1, \loc{1}), (n_2, \loc{2}), \dots , (\sum_{i \in \inter{j,s}} n_{i}, \loc{j}))$-local code is also an $((n_1, \loc{1}), (n_2, \loc{2}), \dots , (n_{s}, \loc{s}))$-local code, because $\loc{1} < \loc{2} < \cdots < \loc{j} < \cdots < \loc{s}$.
\end{proof}

\begin{proof}[Proof of Lemma~\ref{lem_ShorteningSeveral}]
For the  first case, we obtain from~\eqref{eq_DiffLocShortCaseOne} with $\loc{\alpha-1} =\loc{\alpha}-1$:
\begin{align}
& \sum_{i \in  \inter{s-1}} \left \lceil \frac{n'_i}{\loc{i}'+1} \right \rceil \nonumber \\
& = \sum_{\substack{i \in  \inter{s-1}\setminus \\ \{\alpha-1, \alpha \}}} \left \lceil \frac{n_i'}{\loc{i}'+1} \right \rceil + 
\left \lceil \frac{n_{\alpha-1}'}{\loc{\alpha-1}'+1} \right \rceil +  \left \lceil \frac{n_{\alpha}'}{\loc{\alpha}'+1} \right \rceil \label{eq_FirstCasea}.
\end{align}
Inserting the expressions of $n_i'$, $n_{\alpha-1}'$ and $n_{\alpha}'$ as in~\eqref{eq_DiffLocShortCaseOne} into~\eqref{eq_FirstCasea} leads to:
\begin{align}
& =  \sum_{\substack{i \in  \inter{s-1}\setminus \\ \{\alpha-1, \alpha \}}} \left \lceil \frac{n_i}{\loc{i}+1} \right \rceil  + \left \lceil \frac{n_{\alpha-1}+\loc{\alpha}}{\loc{\alpha-1}+1} \right \rceil + \left \lceil \frac{n_{\alpha}-\loc{\alpha}-1}{\loc{\alpha}+1} \right \rceil \nonumber \\
& = \sum_{\substack{i \in  \inter{s-1}\setminus \\ \{\alpha-1, \alpha \}}} \left \lceil \frac{n_i}{\loc{i}+1} \right \rceil + \left \lceil \frac{n_{\alpha-1}}{\loc{\alpha-1}+1} \text{+} \frac{\loc{\alpha}}{\loc{\alpha}} \right \rceil + \left \lceil \frac{n_{\alpha}}{\loc{\alpha}+1} \minusb \frac{\loc{\alpha}+1}{\loc{\alpha}+1} \right \rceil \nonumber \\
&  = \sum_{\substack{i \in  \inter{s-1}}} \left \lceil \frac{n_i}{\loc{i}+1} \right \rceil. \label{eq_FirstCaseRelationOne}
\end{align}
The second relevant term can be expressed as:
\begin{align}
& \sum_{i \in  \inter{s-1}} \loc{i}'  \left \lceil \frac{n'_i}{\loc{i}'+1} \right \rceil \nonumber \\
& = \sum_{\substack{i \in  \inter{s-1}\setminus \\ \{\alpha-1, \alpha \}}} \loc{i}' \left \lceil \frac{n_i'}{\loc{i}'+1} \right \rceil + \loc{\alpha-1}' \left \lceil \frac{n_{\alpha-1}'}{\loc{\alpha-1}'+1} \right \rceil + \loc{\alpha}' \left \lceil \frac{n_{\alpha}'}{\loc{\alpha}'+1} \right \rceil,  \label{eq_FirstCaseb}
\end{align}
and also with~\eqref{eq_DiffLocShortCaseOne} for the expressions of $n_i'$, $n_{\alpha-1}'$ and $n_{\alpha}'$ inserted in~\eqref{eq_FirstCaseb} leads to:
\begin{align}
& = \sum_{\substack{i \in  \inter{s-1}\setminus \\ \{\alpha-1, \alpha \}}} \loc{i} \left \lceil \frac{n_i}{\loc{i}+1} \right \rceil + \loc{\alpha-1} \left \lceil \frac{n_{\alpha-1}+\loc{\alpha}}{\loc{\alpha-1}+1} \right \rceil \nonumber \\
& \qquad + \loc{\alpha} \left \lceil \frac{n_{\alpha}-\loc{\alpha}-1}{\loc{\alpha}+1} \right \rceil \nonumber \\
&  = \sum_{\substack{i \in  \inter{s-1}\setminus \\ \{\alpha-1, \alpha \}}} \loc{i} \left \lceil \frac{n_i}{\loc{i}+1} \right \rceil + (\loc{\alpha}-1) \left \lceil \frac{n_{\alpha-1}}{\loc{\alpha-1}+1} + \frac{\loc{\alpha}}{\loc{\alpha}} \right \rceil \nonumber \\
& \qquad  + \loc{\alpha}\left \lceil \frac{n_{\alpha}}{\loc{\alpha}+1} - \frac{\loc{\alpha}+1}{\loc{\alpha}+1} \right \rceil \nonumber \\
&  = \sum_{\substack{i \in  \inter{s-1}}} \loc{i} \left \lceil \frac{n_i}{\loc{i}+1} \right \rceil - 1. \label{eq_FirstCaseRelationTwo}
\end{align}
For the second case, i.e., $\loc{\alpha-1} \neq \loc{\alpha}-1$, we get for 
\begin{align}
& \sum_{i \in  \inter{s-1} \cup \{ \iota \} } \left \lceil \frac{n'_i}{\loc{i}'+1} \right \rceil \nonumber \\
& = \sum_{\substack{i \in  \inter{s-1}\setminus \\ \{\alpha \}}} \left \lceil \frac{n_i'}{\loc{i}'+1} \right \rceil  + \left \lceil \frac{n_{\alpha}'}{\loc{\alpha}'+1} \right \rceil + \left \lceil \frac{n_{\iota}'}{\loc{\iota}'+1} \right \rceil \label{eq_SecondCasea} 
\end{align}
Inserting the expressions of $n_i'$, $n_{\alpha}'$ and $n_{\iota}'$ as in~\eqref{eq_DiffLocShortCaseTwo} into~\eqref{eq_SecondCasea} leads to:
\begin{align} 
& = \sum_{\substack{i \in  \inter{s-1}\setminus \\ \{\alpha \}}} \left \lceil \frac{n_i}{\loc{i}+1} \right \rceil  + \left \lceil \frac{n_{\alpha} - \loc{\alpha} -1}{\loc{\alpha}+1} \right \rceil + \left \lceil \frac{\loc{\alpha}}{\loc{\alpha}-1+1} \right \rceil \nonumber \\
& = \sum_{\substack{i \in  \inter{s-1}\setminus \\ \{\alpha \}}} \left \lceil \frac{n_i}{\loc{i}+1} \right \rceil  + \left \lceil \frac{n_{\alpha} }{\loc{\alpha}+1} - \frac{\loc{\alpha} + 1}{\loc{\alpha}+1} \right \rceil + \left \lceil \frac{\loc{\alpha}}{\loc{\alpha}} \right \rceil \nonumber \\
& = \sum_{\substack{i \in  \inter{s-1}}} \left \lceil \frac{n_i}{\loc{i}+1} \right \rceil. \label{eq_SecondCaseRelationOne}
\end{align}
The second relevant expression is:
\begin{align}
& \sum_{i \in  \inter{s-1} \cup \{ \iota \} } \loc{i}' \left \lceil \frac{n'_i}{\loc{i}'+1} \right \rceil \nonumber \\
& = \sum_{\substack{i \in  \inter{s \minus 1}\setminus \\ \{\alpha \}}} \loc{i}' \left \lceil \frac{n_i'}{\loc{i}'+1} \right \rceil  + \loc{\alpha}'\left \lceil \frac{n_{\alpha}'}{\loc{\alpha}'+1} \right \rceil + \loc{\iota}' \left \lceil \frac{n_{\iota}'}{\loc{\iota}'+1} \right \rceil, \label{eq_SecondCaseb}
\end{align}
and with~\eqref{eq_DiffLocShortCaseTwo} for the expressions of $n_i'$, $n_{\alpha}'$ and $n_{\iota}'$ inserted in~\eqref{eq_SecondCaseb} leads to:
\begin{align}
& = \sum_{\substack{i \in  \inter{s-1}\setminus \\ \{\alpha \}}} \loc{i} \left \lceil \frac{n_i}{\loc{i}+1} \right \rceil + \loc{\alpha} \left \lceil \frac{n_{\alpha}-\loc{\alpha}-1}{\loc{\alpha}+1} \right \rceil  \nonumber \\
& \qquad + (\loc{\alpha}-1) \left \lceil \frac{\loc{\alpha}}{\loc{\alpha}-1+1} \right \rceil \nonumber \\
& = \sum_{\substack{i \in  \inter{s-1}\setminus \\ \{\alpha \}}} \loc{i} \left \lceil \frac{n_i}{\loc{i}+1} \right \rceil + \loc{\alpha} \left \lceil \frac{n_{\alpha}}{\loc{\alpha}+1} \right \rceil - \loc{\alpha} + (\loc{\alpha}-1) \left \lceil \frac{\loc{\alpha}}{\loc{\alpha}} \right \rceil \nonumber \\
&  = \sum_{\substack{i \in  \inter{s-1}}} \loc{i} \left \lceil \frac{n_i}{\loc{i}+1} \right \rceil - 1. \label{eq_SecondCaseRelationTwo}
\end{align}
The minimum Hamming distance of the shortened code for both cases is with~\eqref{eq_SingletonExp1}:
\begin{equation} \label{eq_ShortenedAllCases}
\begin{split}
d & \leq n' \text{-} k' \text{+} 2 \text{-} \sum_{\substack{i \in  \inter{s \minus 1}}} \left \lceil \frac{n'_i}{\loc{i}'\text{+}1} \right \rceil \text{-} \left \lceil \frac{k'\text{-} \left( \sum  \limits_{i \in \inter{s \minus 1}} \loc{i}' \left \lceil \frac{n_i'}{\loc{i}' \text{+}1} \right \rceil \text{-} 1 \right)}{\loc{s}'} \right \rceil. 
\end{split}
\end{equation}
With~\eqref{eq_FirstCaseRelationOne}, \eqref{eq_FirstCaseRelationTwo}, \eqref{eq_SecondCaseRelationOne} and \eqref{eq_SecondCaseRelationTwo} inserted in~\eqref{eq_ShortenedAllCases} gives:
\begin{align*}
d & \leq n-1 - (k-1) + 2 - \sum_{\substack{i \in  \inter{s-1}}} \left \lceil \frac{n_i}{\loc{i}+1} \right \rceil \\
& \quad - \left \lceil \frac{(k-1) -\left( \sum  \limits_{i \in \inter{s-1}} \loc{i}\left \lceil \frac{n_i}{\loc{i}+1} \right \rceil - 1 \right)}{\loc{s}} \right \rceil \\
& = n \minusb k \text{+} 2 \minusb \sum_{\substack{i \in  \inter{s \minus 1}}} \left \lceil \frac{n_i}{\loc{i} \text{+}1} \right \rceil \minusb \left \lceil \frac{k \minusb \sum \limits_{i \in \inter{s \minus 1}} \loc{i}\left \lceil \frac{n_i}{\loc{i}+1} \right \rceil }{\loc{s}} \right \rceil,
\end{align*}
which is the minimum Hamming distance of the original $\LIN{n}{k}{d}{q}$ $((n_1, \loc{1}),(n_2, \loc{2}),\dots,(n_{s}, \loc{s}))$-local Singleton-optimal ML-LRC.
\end{proof}

\begin{proof}[Proof of Thm.~\ref{thm_AlphabetDiffLocalitiesSeveral}]
It follows from Lemma~\ref{lem_DifferentLocalities} and Lemma~\ref{lem_SeveralEntropySets} for $s$ disjoint sets and the expression as in~\eqref{eq_CMboundSeveralLocalities} follows.
The values of $t_i$ can be bounded by $\min(\lceil n_i/(\loc{i}+1) \rceil , \lfloor (k-1-\prod_{j \in \inter{s-1}\setminus \{i\}} t_jr_j)/\loc{i} \rfloor)$ for all $i \in \inter{s-1}$. We assume that $ \sum_{i \in \inter{s-1}}\loc{i}  \lceil n_i / (\loc{i}+1) \rceil < k-1$ and therefore we obtain~\eqref{eq_CondRunning1Several}. The minimizing value for $t_{s}$ follows from the fact that for $ t_{s} > t_{s}^*$ is greater than $k$ and therefore:
\begin{align*}
t_{s} & \leq \left \lfloor \frac{k-1-\sum \limits_{i \in \inter{s-1}}t_i \loc{i}}{\loc{s}} \right \rfloor.
\end{align*}
\end{proof}
\begin{proof}[Proof of Corollary~\ref{cor_SingletonVSAlphabetSeveral}]
We again first bound the dimension $k_{\text{opt}}^{(q)}$ by the locality-unaware Singleton bound.
\begin{align} 
k  & \leq \min_{t_1,t_2,\dots,t_{s}} \left( \sum_{i \in \inter{s}} t_i \loc{i} \text{+}  k_{\text{opt}}^{(q)}\left(n \text{-} \sum_{i \in \inter{s}} \min(n_i,t_i(\loc{i}\text{+}1)),d \right) \right) \nonumber \\
& = \sum_{i \in \inter{s}} t_i \loc{i} + n - \sum_{i \in \inter{s}} \min(n_i, t_i(\loc{i}+1)) -d+1 \nonumber \\
& = n-d+1-\sum_{i \inter{s}} t_i \label{eq_FirstCMBoundSeveral}.
\end{align}
Assume that the minimum Hamming distance contradicts the Singleton-like bound, i.e.:
\begin{equation*}
 d > n - k + 2 - \sum_{i \in \inter{s \minus 1}} \left \lceil \frac{n_i}{\loc{i}+1} \right \rceil - \left \lceil \frac{k- \sum \limits_{i \in \inter{s \minus 1}} \loc{i}\left \lceil \frac{n_i}{\loc{i}+1} \right \rceil }{\loc{s}} \right \rceil,
\end{equation*}
and inserted into~\eqref{eq_FirstCMBoundSeveral} gives:
\begin{align*} 
& n-d+1-\sum_{i \in \inter{s}} t_i \\ 
& \leq n - n + k - 2 + \sum_{i \in \inter{s \minus 1}} \left \lceil \frac{n_i}{\loc{i}+1} \right \rceil \\
& \quad - \left \lceil \frac{k- \sum \limits_{i \in \inter{s \minus 1}} \loc{i}\left \lceil \frac{n_i}{\loc{i}+1} \right \rceil }{\loc{s}} \right \rceil + 1 -\sum_{i \in \inter{s}} t_i \\
& = k \minusb 1 \text{+} \sum_{i \in \inter{s \minus 1}} \left \lceil \frac{n_i}{\loc{i}+1} \right \rceil \minusb \left \lceil \frac{k \minusb \sum \limits_{i \in \inter{s \minus 1}} \loc{i}\left \lceil \frac{n_i}{\loc{i}+1} \right \rceil }{\loc{s}} \right \rceil \minusb \sum_{i \in \inter{s}} t_i,
\end{align*} 
which is a contradiction for the maximal values of $t_i^*, \forall i \in \inter{s}$ as in~\eqref{eq_CondRunning1Several} and \eqref{eq_CondRunning2Several}.
\end{proof}
\end{document}